\documentclass[letterpaper,11pt]{article}
\usepackage[letterpaper,hmargin=1in,vmargin=1.25in]{geometry}
\usepackage{amsmath}                    
\usepackage{amssymb}                    
\usepackage{amsthm}                     
\usepackage{amsfonts}                   
\usepackage{mathrsfs}                   
\usepackage{thmtools}                    
\usepackage{url}                        
\usepackage{cite}                       
\usepackage{hyperref}                   
\hypersetup{colorlinks=true,linkcolor=[rgb]{0,0,0},citecolor=[rgb]{0,0,0},urlcolor=[rgb]{0,0,0}}
\usepackage{graphicx}                   
\newtheorem{theorem}{Theorem}
\newtheorem{lemma}{Lemma}[section]
\newtheorem{corollary}[lemma]{Corollary}

\newcommand{\ang}[1]{\langle #1\rangle}

\newcommand{\inv}[1]{\frac{1}{#1}}
\renewcommand{\ang}[1]{\langle #1\rangle}
\newcommand{\inner}[2]{\ang{#1,#2}} 
\newcommand{\RE}{\mathbb{R}}            
\newcommand{\eps}{\varepsilon}          
\newcommand{\ST}{\,:\,}                 

\newcommand{\CC}{\mathscr{C}}
\newcommand{\MM}{\mathscr{M}}

\newcommand{\bd}{\partial}
\newcommand{\etal}{\textit{et al.}}
\newcommand{\SP}{\kern+1pt}
\newcommand{\stdpolar}[1]{{#1}^{\circ}} 
\DeclareMathOperator{\diam}{diam}

\DeclareMathOperator{\vol}{vol}
\DeclareMathOperator{\area}{area}

\DeclareMathOperator{\width}{wid}
\DeclareMathOperator{\ray}{ray}
\DeclareMathOperator{\dist}{dist}
\DeclareMathOperator{\base}{base}
\DeclareMathOperator{\shadow}{shadow}

\DeclareMathOperator{\interior}{int}

\newcommand{\pcap}[2]{\textup{cap}_{#1}(#2)}

\begin{document}


\title{Optimal Volume-Sensitive Bounds for Polytope Approximation}

\author{
	Sunil Arya\thanks{Research supported by the Research Grants Council of Hong Kong, China under projects number 16213219 and 16214721.}\\
		Department of Computer Science and Engineering \\
		The Hong Kong University of Science and Technology, Hong Kong\\
		arya@cse.ust.hk \\
		\and
	David M. Mount\\
		Department of Computer Science and 
		Institute for Advanced Computer Studies \\
		University of Maryland, College Park, Maryland \\
		mount@umd.edu \\
}

\date{Revised version for submission to DCG}

\maketitle

\begin{abstract}
Approximating convex bodies is a fundamental question in geometry, which has a wide variety of applications. Given a convex body $K$ in $\RE^d$ for fixed $d$, the objective is to minimize the number of facets of an approximating polytope for a given Hausdorff error $\eps$. It is known that $O((\diam(K)/\eps)^{(d-1)/2})$ facets suffice and are necessary for many instances, such as the Euclidean ball. However, this bound is far from optimal for ``skinny'' convex bodies.

A natural way to characterize the skinniness of a convex object is in terms of its relationship to the Euclidean ball. Given a convex body $K$, its \emph{volume diameter} $\Delta_d(K)$ is defined to be the diameter of a Euclidean ball of the same volume as $K$. The \emph{surface diameter} $\Delta_{d-1}(K)$ is defined analogously for surface area. It follows from generalizations of the isoperimetric inequality that $\diam(K) \geq \Delta_{d-1}(K) \geq \Delta_d(K)$. 

Arya, da Fonseca, and Mount proved that the diameter-based bound could be made sensitive to the surface diameter, improving the above bound to $O((\Delta_{d-1}(K)/\eps)^{(d-1)/2})$. In this paper, we strengthen this by proving the existence of an approximation with $O((\Delta_d(K)/\eps)^{(d-1)/2})$ facets. As a function of volume alone, this bound is tight up to constant factors. 

Our improvements arise from a combination of new ideas. We exploit known properties of the original body and its polar dual. In order to obtain a volume-sensitive bound, we explore the problem of computing a low-complexity polytope that is sandwiched between two given convex bodies. We show that this problem can be reduced to a covering problem involving a natural intermediate body based on the harmonic mean. Our proof relies on a geometric analysis of a relative notion of fatness involving these bodies.
\end{abstract}


\section{Introduction} \label{s:intro}

Approximating convex bodies by polytopes is a fundamental problem which has been extensively studied in the literature (see, e.g., Bronstein~\cite{Bro08}). Given a convex body $K$ in Euclidean $d$-dimensional space and a scalar $\eps > 0$, the problem is to construct a polytope $P$ of low combinatorial complexity that is $\eps$-close to $K$ according to some distance measure. A polytope $P$ is an \emph{$\eps$-approximation} to $K$ if the Hausdorff distance between $K$ and $P$ is at most $\eps$. (Definitions will be provided in Section~\ref{s:prelim}.) The approximation is \emph{outer} if $K \subseteq P$. Throughout, our measure of complexity will be the number of bounding halfspaces (or equivalently, facets) in the approximation, and we assume that the dimension $d$ is a constant. Our asymptotic forms conceal constant factors that depend on $d$.

The approximation bounds presented in the literature are of two common types. In both cases, it is shown that there exists $\eps_0 > 0$ such that the bounds hold for all $\eps \leq \eps_0$. Bounds are said to be \emph{nonuniform} if the value of $\eps_0$ depends on the properties of $K$. Nonuniform bounds, such as those appearing in the works of Gruber~\cite{Gru93a}, Clarkson~\cite{Cla06}, and others \cite{Bor00,Sch87,McV75,Tot48}, typically hold subject to smoothness conditions on $K$'s boundary. 

In contrast, if the value of $\eps_0$ is independent of $K$ (but may depend on $d$), the bound is said to be \emph{uniform}. Such bounds hold without any additional smoothness assumptions. As an example of such a bound, Dudley~\cite{Dud74} showed that any convex body $K$ can be $\eps$-approximated by a polytope $P$ with at most $c_d \cdot (\diam(K)/\eps)^{(d-1)/2}$ facets, where $c_d$ is a constant depending on the dimension, $\diam(K)$ denotes $K$'s diameter, and $0 < \eps \leq \diam(K)$. (A simple self-contained proof was given by Har-Peled and Jones~\cite{HaJ24}.) Bronshteyn and Ivanov~\cite{BrI76} showed that the same bound holds for the number of vertices. These results have numerous applications in computational geometry, for example, in the construction of coresets~\cite{AHV05, ArC14, AFM17b}. Our bounds are of the uniform type.

The approximation bounds of both Dudley and Bronshteyn-Ivanov are tight in the worst case up to constant factors, specifically when $K$ is a Euclidean ball (see, e.g., \cite{Bro08}). However, these bounds may be significantly suboptimal if $K$ is ``skinny''. A natural way to characterize the skinniness of a convex object is in terms of its relationship to the Euclidean ball. The \emph{volume diameter} of a convex body $K$ in $\RE^d$, denoted $\Delta_d(K)$, is defined to be the diameter of a Euclidean ball of the same volume as $K$, or equivalently,
\[
    \Delta_d(K)
        ~ = ~ 2 \left( \frac{\vol_d(K)}{\vol_d(B^d_2)} \right)^{\kern-2pt 1/d},
\]
where $B^d_2$ is the Euclidean unit ball in $\RE^d$, and $\vol_d(\cdot)$ denotes the $d$-dimensional Lebesgue measure.

The \emph{surface diameter}, denoted $\Delta_{d-1}(K)$, is defined analogously based on the surface areas of $K$ and a unit ball. These quantities are closely related to the classical concepts of \emph{quermassintegrals} and of \emph{intrinsic volumes} of the convex body \cite{McM75,McM91}. From generalizations of the isoperimetric inequality it follows that $\diam(K) \geq \Delta_{d-1}(K) \geq \Delta_d(K)$~\cite{McM91}. 

The question considered in this paper is whether it is possible to strengthen Dudley's bound by expressing the complexity of the approximation in terms of a body's volume diameter. Arya, da Fonseca, and Mount~\cite{AFM12b} proved that Dudley's bound could be made surface-area sensitive, improving the bound to $O\big( (\Delta_{d-1}(K)/\eps)^{(d-1)/2} \big)$. In this paper, we strengthen this further to produce a volume-sensitive bound. 

Before stating our result, we need to address an issue arising with extremely thin bodies. Suppose that we have a Euclidean ball in $\RE^{d-1}$. By the tightness of Dudley's bounds, any approximating polytope requires significant complexity. However, if we extrude this body infinitesimally into the next higher dimension (imagine a large circular disk cut from a thin sheet of paper in $\RE^3$), its volume diameter $\Delta_d(K)$ can be arbitrarily small. To deal with such degenerate cases, we require that the body have a minimum width of at least $\eps$ along any direction. Alternatively, we could fatten the body by taking the Minkowski sum with a Euclidean ball of radius $\eps$ before taking the approximation.

\begin{theorem} \label{thm:main}
Consider any convex body $K$ in $\RE^d$ and any $\eps > 0$ such that the width of $K$ in any direction is at least $\eps$. There exists an outer $\eps$-approximating polytope $P$ for $K$ whose number of facets is at most
\[
	c_d \left(\frac{\Delta_d(K)}{\eps} \right)^{\kern-2pt\frac{d-1}{2}},
\]
where $c_d$ is a constant (depending on $d$).
\end{theorem}

This matches Dudley's bound for fat objects, such as Euclidean balls. In contrast, the volume diameter of a skinny pencil-like object that has width $\eps$ along $d-1$ dimensions and width $1$ along one dimension has a volume diameter of only $\eps^{(d-1)/d}$, and Theorem~\ref{thm:main} yields a bound of only $O\big( 1/\eps^{(d-1)/2d} \big)$, which improves Dudley's bound by a factor of roughly $1/\eps^{(d-2)/2}$.

As a function of volume alone, the bound of Theorem~\ref{thm:main} is tight up to constant factors. To see why, observe that the bound can be stated in terms of $K$'s volume as $c_d \cdot \vol_d(K)^{(d-1)/2d} / \eps^{(d-1)/2}$. Clearly, $\vol_d(K) \leq \diam(K)^d$, and the tightness of Dudley's bound implies that, up to constant factors, the number of facets needed is at least $(\diam(K)/\eps)^{(d-1)/2} \geq \vol_d(K)^{(d-1)/2d} / \eps^{(d-1)/2}$.

This bound is the strongest to date as a function of intrinsic volumes. To contrast this with the area-sensitive bound of \cite{AFM12b}, consider a pancake-like object $K$ that has width $\eps$ along one dimension and unit width along all the others. This body has volume $\Theta(\eps)$ and surface area $\Theta(1)$, and therefore $\Delta_d(K) = \Theta(\eps^{1/d})$ and $\Delta_{d-1}(K) = \Theta(1)$. The area-sensitive bound matches Dudley's bound, while the volume-sensitive bound is better by a factor of $1/\eps^{(d-1)/2d}$. The problem of shape-sensitive approximations was also studied by Bonnet~\cite{Bon18}, but his results can at best be used to obtain area-sensitive bounds.

Our improvements arise from a combination of new ideas. As in earlier works, we employ the use of covers based on Macbeath regions together with known properties of the original body and its polar dual, in particular, the Mahler volume and the Blaschke–Santal\'{o} inequality. In order to obtain a volume-sensitive bound, we explore the problem of computing a low-complexity polytope that is sandwiched between two given convex bodies. We show that this problem can be reduced to a covering problem involving a natural intermediate body based on the harmonic mean. Our proof relies on a geometric analysis of a relative notion of fatness involving these bodies.

The remainder of the paper is organized as follows. In Section~\ref{s:techniques}, we give a high-level overview of our methods. In Section~\ref{s:prelim}, we present basic definitions and concepts that will be used throughout the paper, including a central concept, called relative fatness, and we introduce two convex bodies, the arithmetic-mean and harmonic-mean bodies. In Sections~\ref{s:mnet-sizes} and~\ref{s:hm-fat}, we explore the relevant properties of these bodies and analyze the sizes of the covers that form the basis of our approximation. In Section~\ref{s:hausdorff}, we combine these elements to derive the final approximation. Finally, in Section~\ref{s:nonuniform} we present an additional result, a relatively simple derivation of a volume-sensitive approximation bound in the nonuniform setting.

\section{Overview of Techniques} \label{s:techniques}

The problem of approximating a convex body by a polytope can be reduced to ``sandwiching'' a polytope between two nested convex bodies, denoted $K_0$ and $K_1$. For example, in the case of an $\eps$-approximation to $K$ in the Hausdorff distance, we could define $K_0 = K$ and $K_1$ as the Minkowski sum of $K$ with a Euclidean ball of radius $\eps$, that is, $K_1 = K \oplus \eps B^d_2$ (see Figure~\ref{f:prelim}(a)). Much of the previous work in this area has focused on the specific manner in which $K_1$ is defined relative to $K_0$.

Recent approaches to convex approximation have been based on covering the body to be approximated with convex objects that respect the local shape of the body being approximated (see, e.g.,~\cite{AAFM22,AFM24}). Macbeath regions have been a key tool in this regard. Consider a convex body $K$ and a point $x$ in the interior of $K$. Intuitively, the \emph{Macbeath region} at $x$, denoted $M_K(x)$, is the largest centrally symmetric body nested within $K$ and centered at $x$ (see Section~\ref{s:mac-prop} for definitions). A Macbeath region that has been shrunk by some constant factor $\lambda$ is denoted by $M_K^{\lambda}(x)$ (see Figure~\ref{f:prelim}(b)). Shrunken Macbeath regions have nice packing and covering properties, and they behave much like metric balls. 

\begin{figure}[htbp]
    \centering\includegraphics[scale=0.8]{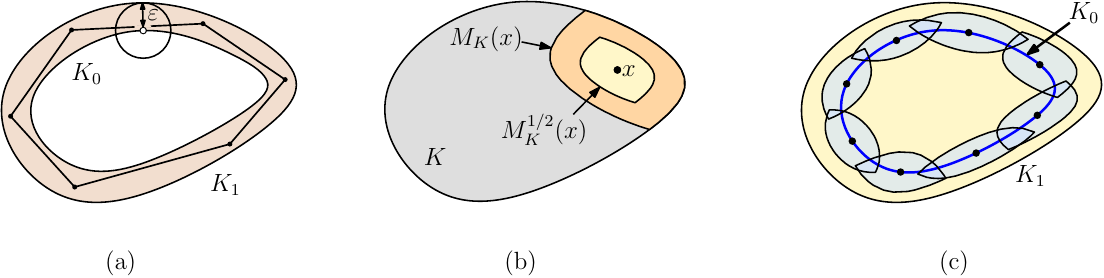}
    \caption{(a) Hausdorff approximation, (b) Macbeath regions and (c) covering the boundary of $K_0$ by Macbeath regions.} \label{f:prelim} 
\end{figure}

A natural way to construct a sandwiching polytope between two nested bodies $K_0$ and $K_1$ is to form a collection of shrunken Macbeath regions with respect to $K_1$ that cover the boundary of $K_0$ (see Figure~\ref{f:prelim}(c)). If done properly, a sandwiching polytope can be constructed by sampling a constant number of points from each of these Macbeath regions and taking the convex hull of their union. Thus, up to constant factors, the number of Macbeath regions provides an upper bound on the number of vertices in the sandwiched polytope. The concept of an MNet (defined in Section~\ref{s:mac-prop}) will be useful to characterize a set of Macbeath regions that covers a portion of a convex body. 

The principal challenge is to determine the number of Macbeath regions needed to form such a cover. This is often done through a packing argument. The objective is to show that each Macbeath region covers a significant fraction of volume of the convex body being covered. To obtain the best bounds, Macbeath regions may be constructed in either the original body or in its polar body. This is because of a well-known concept from convexity called the Mahler volume (defined later in Section~\ref{s:centrality}), which states that the product of the volume of a convex body and its polar is bounded from below by a constant. When objects are fat, it can be shown that small Macbeath regions in the original body correspond to large Macbeath regions in the polar body. Hence, the packing argument can be pushed through from either the original side or the polar side.

However, in a general context this correspondence cannot be used as a basis for a packing argument. To see the issue, consider the two bodies $K_0$ and $K_1$ shown in Figure~\ref{f:rel-fat}(a), where $K_0$ is a diamond shape nested within the square $K_1$. Consider a $\frac{1}{2}$-scaled Macbeath region centered at a point $x$ that lies at the top vertex of $K_0$. Observe that almost all of its volume lies outside of $K_0$. This is problematic for a packing argument, since such a body covers very little of the volume of $K_0$. 

\begin{figure}[htbp]
    \centering\includegraphics[scale=0.8,page=1]{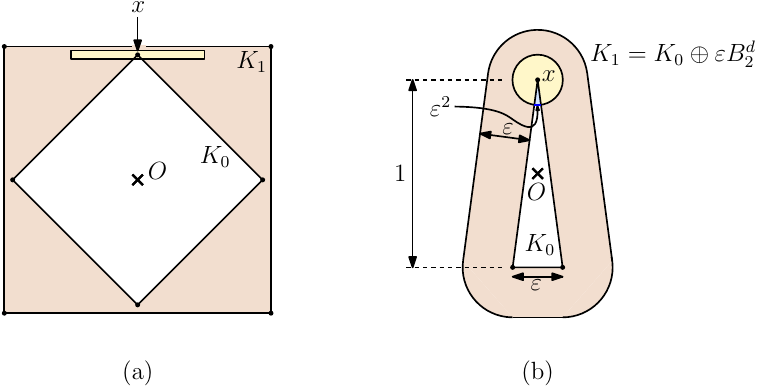}
    \caption{Relative fatness.} \label{f:rel-fat} 
\end{figure}

Intuitively, while the body $K_0$ is ``fat'' in the standard sense%
\footnote{A convex body in $\RE^d$ is \emph{fat} if it can be sandwiched between two Euclidean balls whose radii differ by a factor depending only on $d$.},
it is not fat ``relative'' to the enclosing body $K_1$. Although this example is not typical of approximations (where $K_1$ is an $\eps$-expansion of $K_0$), it is not hard to create more typical examples where this same phenomenon arises (see Figure~\ref{f:rel-fat}(b)). These difficulties are further enhanced by the fact that relative fatness also needs to hold in the polar setting.

Our overall approach will be structurally similar to previous Macbeath-based constructions, but the analysis is complicated due to the additional considerations arising from relative fatness. We introduce an intermediate body that is sandwiched between $K_0$ and $K_1$, called the harmonic-mean body (Section~\ref{s:am-hm}). We next extend results from our earlier paper \cite{AFM24} to bound the sizes of MNets in our new setting (Section~\ref{s:mnet-sizes}), and we show that the inner body is relatively fat with respect to the harmonic-mean body (Section~\ref{s:hm-fat}). Finally, we show that the approximation can be performed using the harmonic-mean body, and we will combine these elements to obtain the final construction of the $\eps$-approximation (Section~\ref{s:hausdorff}).

\section{Preliminaries} \label{s:prelim}

In this section, we recall some standard notation and concepts. Throughout, $K$ denotes a convex body in $\RE^d$, that is, a compact convex subset with a nonempty interior, and $\eps$ denotes a fixed approximation parameter. Let $\bd K$ denote its boundary. Let $\vol(K) = \vol_d(K)$ denote its $d$-dimensional Lebesgue measure, and let $\area(K) = \vol_{d-1}(\bd K)$ denote its surface area. For $\alpha \geq 0$, $\alpha K$ denotes a uniform scaling of $K$ about the origin, and for $x \in \RE^{d}$, $K + x$ denotes the translation of $K$ by $x$. Given a convex body $L$, let $K \oplus L$ denote the Minkowski sum of $K$ and $L$, that is, $\{x + y \ST x \in K, y \in L\}$. Let $B^d_2$ denote the Euclidean ball of unit radius centered at the origin.

Throughout, we use $\inner{\cdot}{\cdot}$ to denote the standard inner (dot) product and use $\|\cdot\| = \sqrt{\inner{\cdot}{\cdot}}$ to denote the Euclidean norm. Given two convex bodies $K$ and $L$ in $\RE^d$, their \emph{Hausdorff distance} is defined to be
\[
    \min \left\{ r \geq 0 \ST K \subseteq L \oplus r B^d_2 ~\text{and}~ L \subseteq K \oplus r B^d_2) \right\}.
\]
Given a unit vector $u$, the \emph{width of $K$ in direction $u$} is the smallest distance between two hyperplanes, both orthogonal to $u$, that enclose $K$. This is equivalent to the absolute difference between the values of the support function of $K$ in $u$ and $-u$. The \emph{minimum width} of $K$ is the minimum over all directional widths.

To avoid specifying many real-valued constants that arise in our analysis, we will often hide them using asymptotic notation. For a positive real $x$, we use the notation $O(x)$ to denote a quantity whose value is at most $c \SP x$ for an appropriately chosen constant $c$. We use $\Omega(x)$ for a quantity that is at least $c \SP x$. We use $\Theta(x)$ to denote a quantity that lies within the interval $[c \SP x, c' x]$, for appropriate constants $c$ and $c'$. Hidden constants may depend on the dimension $d$, but they do not depend on $K$ or $\eps$. 

\subsection{Polarity and Centrality Properties} \label{s:centrality}

Given a convex body $K \subseteq \RE^d$ that contains the origin $O$ in its interior, define its \emph{polar}, denoted $\stdpolar{K}$, to be the convex set
\[
	\stdpolar{K}
		~ = ~ \{ u \ST \inner{u}{v} \leq 1, \hbox{~for all $v \in K$} \}.
\]
The polar has many useful properties (see, e.g., Eggleston~\cite{Egg58}). For example, it is well known that $\stdpolar{K}$ is bounded and $\stdpolar{(\stdpolar{K})} = K$. Furthermore, if $K_1$ and $K_2$ are two convex bodies that contain the origin such that $K_1 \subseteq K_2$, then $\stdpolar{K}_2 \subseteq \stdpolar{K}_1$. 

Given a nonzero vector $v \in \RE^d$, we define its ``polar'' $\stdpolar{v}$ to be the hyperplane that is orthogonal to $v$ and at distance $1/\|v\|$ from the origin, on the same side of the origin as $v$. The polar of a hyperplane is defined as the inverse of this mapping. We may equivalently define $\stdpolar{K}$ as the intersection of the closed halfspaces that contain the origin, bounded by the hyperplanes $\stdpolar{v}$, for all $v \in K$. 

Given a convex body $K$ that contains the origin in its interior, define its \emph{Mahler volume} to be $\vol(K) \cdot \vol(\stdpolar{K})$. The Mahler volume has been well studied (see, e.g.~\cite{San49,MeP90,Sch93}). It is invariant under linear transformations but depends on the location of the origin. The famous Blaschke–Santal\'{o} inequality states that if the centroid of $K$ coincides with the origin, then its Mahler volume is bounded above by some constant depending only on the dimension (see, e.g., \cite{BoM87,Kup08,Naz12}). Throughout the paper, we fix a suitably large constant $c_0$ (depending on the dimension), and we say that $K$ is \emph{well-centered} if its Mahler volume is at most $c_0$. By known results, we have the following.

\begin{lemma} \label{lem:mahler-bounds}
Given a convex body $K \subseteq \RE^d$ whose interior contains the origin,
\begin{enumerate}\setlength{\itemsep}{-0.5ex}\setlength{\parsep}{0pt}
\item[$(i)$] $\vol(K) \cdot \vol(\stdpolar{K}) = \Omega(1)$,
\item[$(ii)$] if $K$'s centroid coincides with the origin, then $K$ is well-centered, that is, $\vol(K) \cdot \vol(\stdpolar{K}) = O(1)$.
\end{enumerate}
\end{lemma}

\subsection{Caps, Rays, and Relative Measures} \label{s:cap-prop}

Consider a compact convex body $K$ in $d$-dimensional space $\RE^d$ with the origin $O$ in its interior. A \emph{cap} $C$ of $K$ is defined to be the nonempty intersection of $K$ with a halfspace. Letting $h_1$ denote a hyperplane that does not pass through the origin, let $\pcap{K}{h_1}$ denote the cap resulting by intersecting $K$ with the halfspace bounded by $h_1$ that does not contain the origin (see Figure~\ref{f:widray}(a)). Define the \emph{base} of $C$, denoted $\base(C)$, to be $h_1 \cap K$. Letting $h_0$ denote a supporting hyperplane for $K$ and $C$ parallel to $h_1$, define an \emph{apex} of $C$ to be any point of $h_0 \cap K$.

\begin{figure}[htbp]
    \centering\includegraphics[scale=0.8]{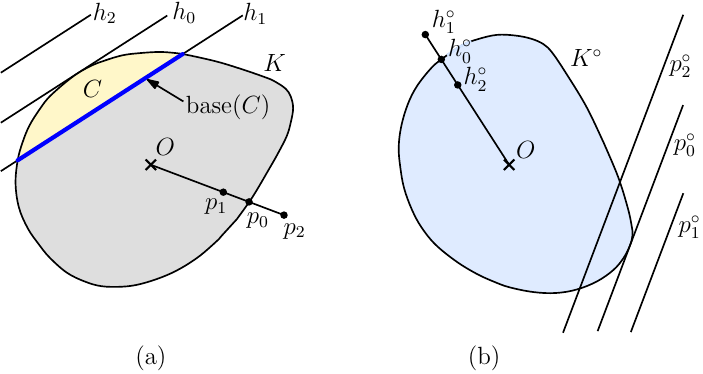}
    \caption{Convex body $K$ and polar $\stdpolar{K}$ with definitions used for width and ray.} \label{f:widray} 
\end{figure}

We define the \emph{absolute width} of cap $C$ to be $\dist(h_1,h_0)$. When a cap does not contain the origin, it will be convenient to define distances in relative terms. We define the \emph{relative width} of such a cap $C$, denoted $\width_K(C)$, to be the ratio $\dist(h_1,h_0) / \dist(O,h_0)$ and, to simplify notation, define $\width_K(h_1) = \width_K(\pcap{K}{h_1})$. Observe that as a hyperplane is translated from a supporting hyperplane to the origin, the relative width of its cap ranges from $0$ to a limiting value of $1$.

We also characterize the proximity of a point to the boundary in both absolute and relative terms. Given a point $p_1 \in K$, let $p_0$ denote the point of intersection of the ray $O p_1$ with the boundary of $K$. Define the \emph{absolute ray distance} of $p_1$ to be $\|p_1 p_0\|$, and define the \emph{relative ray distance} of $p_1$, denoted $\ray_K(p_1)$, to be the ratio $\|p_1 p_0\| / \|O p_0\|$. Relative widths and relative ray distances are both affine invariants. Throughout the paper, unless otherwise specified, \emph{widths and ray distances are understood to be relative}.

We can also define volumes in an affine invariant manner. Recall that $\vol(\cdot)$ denotes the standard Lebesgue volume measure. For any region $\Lambda \subseteq K$, define the \emph{relative volume} of $\Lambda$ with respect to $K$, denoted $\vol_K(\Lambda)$, to be $\vol(\Lambda)/\vol(K)$.

With the aid of the polar transformation, we can extend the concepts of width and ray distance to objects lying outside of $K$. Consider a hyperplane $h_2$ parallel to $h_1$ that lies beyond the supporting hyperplane $h_0$ (see Figure~\ref{f:widray}(a)). It follows that $\stdpolar{h}_2 \in \stdpolar{K}$, and we define $\width_K(h_2) = \ray_{\stdpolar{K}}(\stdpolar{h}_2)$ (see Figure~\ref{f:widray}(b)). Similarly, for a point $p_2 \notin K$ that lies along the ray $O p_1$, it follows that the hyperplane $\stdpolar{p}_2$ intersects $\stdpolar{K}$, and we define $\ray_K(p_2) = \width_{\stdpolar{K}}(\stdpolar{p}_2)$. By properties of the polar transformation, it is easy to see that $\width_K(h_2) = \dist(h_0,h_2) / \dist(O,h_2)$. Similarly, $\ray_K(p_2) = \|p_0 p_2\| / \|O p_2\|$. Henceforth, we will omit references to $K$ when it is clear from context.

Some of our results apply only when we are sufficiently close to the boundary of $K$. Given $\alpha \leq \frac{1}{2}$, we say that a cap $C$ is \emph{$\alpha$-shallow} if $\width(C) \leq \alpha$, and we say that a point $p$ is \emph{$\alpha$-shallow} if $\ray(p) \leq \alpha$. We will simply say \emph{shallow} to mean $\alpha$-shallow, where $\alpha \leq \frac{1}{2}$ is a sufficiently small constant.

We state some useful technical results regarding ray distances and cap widths. The missing proofs can be found in~\cite[Section~2.3]{AFM24}.

\begin{lemma} \label{lem:raydist-width}
Let $C$ be a cap of $K$ that does not contain the origin, and let $p$ be a point in $C$. Then $\ray(p) \leq \width(C)$.
\end{lemma}

There are two natural ways to associate a cap with any point $p \in K$. The first is the \emph{minimum volume cap}, which is any cap whose base passes through $p$ and that has minimum volume among all such caps. For the second, assume that $p \neq O$, and let $p_0$ denote the point of intersection of the ray $O p$ with the boundary of $K$. Let $h_0$ be any supporting hyperplane of $K$ at $p_0$. Take the cap $C$ induced by a hyperplane parallel to $h_0$ passing through $p$. As stated in the following lemma, this is the cap of minimum width containing $p$.

\begin{lemma}
\label{lem:min-width-cap}
For any $p \in K \setminus \{O\}$, consider the cap $C$ defined above. Then $\width(C) = \ray(p)$ and further, $C$ has the minimum width over all caps that contain $p$.
 \end{lemma}

The next lemma shows that cap widths behave nicely under containment.

\begin{lemma} \label{lem:cap-containment-width}
Let $C_1$ and $C_2$ be two caps that do not contain the origin such that $C_1 \subseteq C_2$. Then $\width(C_1) \leq \width(C_2)$.
\end{lemma}

\begin{proof}
Consider the intersection point $p$ of the base of $C_1$ with the ray joining $O$ to the apex of $C_1$. By Lemma~\ref{lem:min-width-cap} and the remarks preceding it, $\ray(p) = \width(C_1)$. Since $C_1 \subseteq C_2$, it follows that $p \in C_2$, and so by Lemma~\ref{lem:raydist-width}, $\ray(p) \leq \width(C_2)$. The lemma follows.
\end{proof}

Given any cap $C$ and a real $\lambda > 0$, we define its $\lambda$-expansion, denoted $C^{\lambda}$, to be the cap of $K$ cut by a hyperplane parallel to the base of $C$ such that the absolute width of $C^{\lambda}$ is $\lambda$ times the absolute width of $C$. (Notice that the expansion of a cap may contain the origin, and, indeed, if the expansion is large enough, it may be the same as $K$.) An easy consequence of convexity is that, for $\lambda \geq 1$, $C^{\lambda}$ is a subset of the region obtained by scaling $C$ by a factor of $\lambda$ about its apex. This implies the following lemma.

\begin{lemma} \label{lem:cap-exp}
Given any cap $C$ and a real $\lambda \geq 1$, $\vol(C^{\lambda}) \leq \lambda^d \cdot \vol(C)$.
\end{lemma}

\subsection{Macbeath Regions and MNets} \label{s:mac-prop}

Given a convex body $K$, a point $x \in K$, and a scaling factor $\lambda > 0$, the \emph{Macbeath region} $M_K^\lambda(x)$ is defined as
\[
    M_K^\lambda(x) 
        ~ = ~ x + \lambda ((K - x) \cap (x - K)).
\]
It is easy to see that $M_K^1(x)$ is the intersection of $K$ with the reflection of $K$ around $x$, and so $M_K^1(x)$ is centrally symmetric about $x$. In fact, it is the largest centrally symmetric body centered at $x$ and contained in $K$. Furthermore, $M_K^\lambda(x)$ is a copy of $M_K^1(x)$ scaled by the factor $\lambda$ about the center $x$ (see Figure~\ref{f:prelim}(a)). We will omit the subscript $K$ when the convex body is clear from the context. For convenience, we define $M(x) = M^1(x)$. 

In the following, we summarize a number of important properties of Macbeath regions and MNets (defined below). The missing proofs can be found in~\cite[Section~2.5]{AFM24} unless otherwise indicated. A good source of general information about Macbeath regions can be found in the survey by B{\'a}r{\'a}ny~\cite{Bar00}.

The first lemma shows that if two shrunken Macbeath regions overlap, then a constant factor expansion of one encloses the other. In this sense, the Macbeath regions that overlap can serve as proxies for each other~\cite{BCP93,ELR70,AFM17a}. The second lemma relates a Macbeath region and any cap containing its center.

\begin{lemma} \label{lem:mac-mac}
Let $K$ be a convex body and let $\lambda \leq \frac{1}{5}$ be any real. If $x, y \in K$ such that $M^{\lambda}(x) \cap M^{\lambda}(y) \neq \emptyset$, then $M^{\lambda}(y) \subseteq M^{4\lambda}(x)$.
\end{lemma}

\begin{lemma} \label{lem:mac-cap-var}
Let $K$ be a convex body and $\lambda > 0$. If $x$ is a point in a cap $C$ of $K$, then $M^\lambda(x) \cap K \subseteq C^{1+\lambda}$.
\end{lemma}

The next three lemmas relate the volume of caps and associated Macbeath regions. 

\begin{lemma}[B{\'a}r{\'a}ny~\cite{Bar07}] \label{lem:min-vol-cap1}
Given a convex body $K \subseteq \RE^d$, let $C$ be a $\frac{1}{3}$-shallow cap of $K$, and let $p$ be the centroid of $\base(C)$. Then $C \subseteq M^{2d}(p)$.
\end{lemma}

\begin{lemma}
\label{lem:wide-cap}
Let $0 < \beta < 1$ be any constant. Let $K \subseteq \RE^d$ be a well-centered convex body, $p \in K$, and $C$ be the minimum volume cap associated with $p$. If $C$ contains the origin or $\width(C) \geq \beta$, then $\vol_K(M(p)) = \Omega(1)$.
\end{lemma}

\begin{lemma} \label{lem:min-vol-cap2}
Given a convex body $K \subseteq \RE^d$, let $C$ be a $\frac{1}{3}$-shallow cap of $K$, and let $p$ be the centroid of $\base(C)$. Then $\vol(M(p)) = \Theta(\vol(C))$.
\end{lemma}

In the next lemma, we show that the width of the minimum volume cap for $p$ is within a constant factor of the ray distance of $p$.

\begin{lemma} \label{lem:min-vol-cap3} 
Let $K$ be a convex body, $p \in K$, and $C$ be the minimum volume cap associated with $p$. If $C$ is $\frac{1}{3}$-shallow, then $\width(C) \leq (2d+1) \cdot \ray(p)$.
\end{lemma}

\begin{proof}
We may assume that $\ray(p) \leq 1/(3(2d+1))$, since otherwise the lemma holds trivially. By a well-known property of minimum volume caps, $p$ is the centroid of the base of $C$~\cite{ELR70}. By Lemma~\ref{lem:min-vol-cap1}, we have $C \subseteq M^{2d}(p)$. By definition, $C \subseteq K$, and so $C \subseteq M^{2d}(p) \cap K$. Applying Lemma~\ref{lem:mac-cap-var} to point $p$ and the minimum width cap $W$ for $p$, we have $M^{2d}(p) \cap K \subseteq W^{2d+1}$. Thus, $C \subseteq W^{2d+1}$. By Lemma~\ref{lem:min-width-cap}, $\width(W) = \ray(p)$, and so $\width(W^{2d+1}) = (2d+1) \ray(p)\leq 1/3$. Since $C$ and $W^{2d+1}$ are both $(1/3)$-shallow, and $C \subseteq W^{2d+1}$, it follows from Lemma~\ref{lem:cap-containment-width} that $\width(C) \leq \width(W^{2d+1})$. Thus $\width(C) \leq (2d+1) \ray(p)$, as desired.
\end{proof}

The next lemma states lower and upper bounds on the relative volume of a Macbeath region based on the width of the associated cap or the ray distance of its center.

\begin{lemma} \label{lem:vol-mac-bounds}
Let $\eps > 0$ be sufficiently small and let $K \subseteq \RE^d$ be a well-centered convex body. Then:
\begin{enumerate}
\item[$(i)$] Let $M$ be a Macbeath region centered at the centroid of the base of a cap $C \subseteq K$ of width $\eps$. Then $\vol_K(M) = O(\eps)$ and $\vol_K(M) = \Omega(\eps^d)$.
\item[$(ii)$] Let $M$ be a Macbeath region centered at a point $x \in K$ whose ray distance is $\eps$. Then $\vol_K(M) = O(\eps)$ and $\vol_K(M) = \Omega(\eps^d)$.
\end{enumerate}
\end{lemma}

\begin{proof}
By Lemma~\ref{lem:mac-cap-var}, $M \subseteq C^2$. Also, $C^2 \subseteq S_K$, where $S_K = K \setminus (1-2\eps)K$. Thus
\[
    \vol_K(M) 
        ~ \leq ~ \vol_K(C^2) 
        ~ \leq ~ \vol_K(S_K) 
        ~ =    ~ 1 - (1-2\eps)^d 
        ~ =    ~ O(\eps).
\]
Similarly, in part (ii), by considering the cap defined by the supporting hyperplane of $K(1-\eps)$ at $x$, we can show that $\vol_K(M) = O(\eps)$.

Next, we show the lower bound on $\vol_K(M)$ in part (i). Let $y$ denote the point $\psi(C) \in \stdpolar{K}$ and let $M'$ denote the Macbeath region $M^{1/5}_{\stdpolar{K}}(y)$. Note that $\ray_{\stdpolar{K}}(y) = \eps$. As the cap $C$ and Macbeath region $M'$ satisfy the conditions of Lemma~\ref{lem:mahler-mac}, we have $\vol_K(C) \cdot \vol_{\stdpolar{K}}(M') = \Omega(\eps^{d+1})$. By the upper bound in (ii), we have $\vol_{\stdpolar{K}}(M') = O(\eps)$ and thus $\vol_K(C) = \Omega(\eps^d)$. Also, by Lemma~\ref{lem:min-vol-cap2}, we have $\vol(M) = \Omega(\vol(C)$. Thus $\vol_K(M) = \Omega(\eps^d)$, as desired. Similarly, we can establish the lower bound in part (ii).
\end{proof}

The following lemma states that points in a shrunken Macbeath region all have similar ray distances. 

\begin{lemma} \label{lem:core-ray}
Let $K$ be a convex body. If $x$ is a $\frac{1}{2}$-shallow point in $K$ and $y \in M^{1/5}(x)$, then $\ray(x)/2 \leq \ray(y) \leq 2 \cdot \ray(x)$.
\end{lemma}

The next lemma shows that translated copies of a Macbeath region act as proxies for Macbeath regions in the vicinity.

\begin{lemma} \label{lem:mac-trans}
Given $0 \leq \lambda \leq \frac{1}{2}$, $\gamma \geq 0$, and a convex body $K$, let $x \in K$ and let $R = M(x)-x$. Let $y$ be a point in $x + \lambda R$. Then $y + \gamma R \subseteq M^{2\gamma}(y)$.
\end{lemma}

\begin{proof}
Treating $x$ as the origin, we have $R \subseteq K$ and $y \in \lambda R$. It follows that $y + (1-\lambda)R \subseteq K$. Recall that the Macbeath region $M(y)$ is the maximal centrally symmetric convex body centered at $y$ and contained within $K$. Thus, $y + (1-\lambda) R \subseteq M(y)$. This implies that $y + \gamma R \subseteq M^{\frac{\gamma}{1-\lambda}}(y) \subseteq M^{2\gamma}(y)$.
\end{proof}

We employ Macbeath region-based coverings in our polytope approximation scheme. In particular, we employ the concept of MNets, as defined in~\cite{AFM24}. Let $K \subseteq \RE^d$ be a convex body, let $\Lambda$ be an arbitrary subset of $\interior(K)$, and let $c \geq 5$ be any constant. Given $X \subseteq K$, define $\MM_K^{\lambda}(X) = \{ M_K^{\lambda}(x) : x \in X\}$. Define a \emph{$(K, \Lambda,c)$-MNet} to be any maximal set of points $X \subseteq \Lambda$ such that the shrunken Macbeath regions $\MM_K^{1/4c}(X)$ are pairwise disjoint (see Figure~\ref{fig:mnet}(a) and~(b)). We refer to $c$ as the expansion factor of the MNet. The following lemma, proved in~\cite[Section~4]{AFM24}, summarizes the key properties of MNets.

\begin{figure}[htbp]
    \centering\includegraphics[scale=0.8]{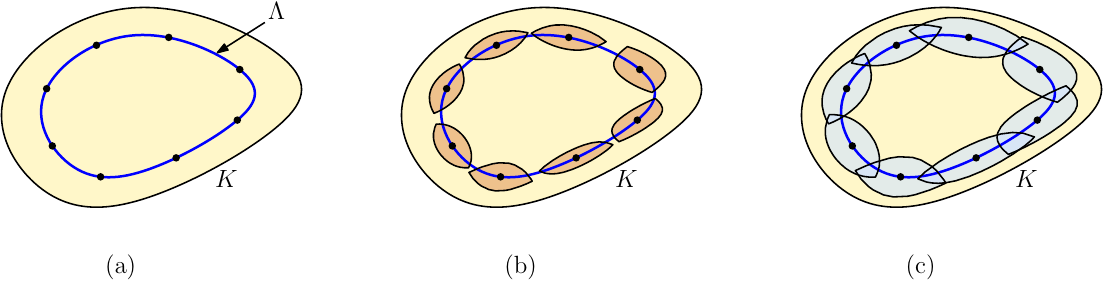}
    \caption{An MNet for a set $\Lambda$ within a convex body $K$. (Not drawn to scale.)} \label{fig:mnet}
    \end{figure}

\begin{lemma}\label{lem:delone}
Given a convex body $K \subseteq \RE^d$, $\Lambda \subset \interior(K)$, and $c \geq 5$, a $(K,\Lambda,c)$-MNet $X$ satisfies the following properties:
\begin{itemize}
    \item (Packing) The elements of $\MM_K^{1/4c}(X)$ are pairwise disjoint (see Figure~\ref{fig:mnet}(b)).
    \item (Covering) The union of $\MM_K^{1/c}(X)$ covers $\Lambda$ (see Figure~\ref{fig:mnet}(c)).
    \item (Buffering) The union of $\MM_K(X)$ is contained within $K$.
\end{itemize}
\end{lemma}

For the purposes of this paper, $c$ will be any sufficiently large constant, specifically $c \geq 5$. To simplify notation, we use $(K,\Lambda)$-MNet to refer to such an MNet.

As mentioned before, we reduce our polytope approximation problem to that of finding a polytope that is sandwiched between two convex bodies. In turn, we tackle this problem using MNets as indicated in the next lemma. A similar result was proved in~\cite{AFM24} under the more restrictive assumption that $K_1$ is a scaled copy of $K_0$.

\begin{lemma} \label{lem:MNet-approx}
Let $K_0 \subset K_1$ be two convex bodies. Let $X$ be a $(K_1,\partial K_0)$-MNet. Then there exists a polytope $P$ with $O(|X|)$ vertices such that $K_0 \subseteq P \subseteq K_1$.
\end{lemma}

\begin{proof}
Define a \emph{half-ellipsoid} to be the intersection of an ellipsoid with a halfspace whose bounding hyperplane passes through its center. Let $c$ be the expansion factor of the MNet $X$. For each Macbeath region $M \in \MM^{1/c}(X)$, choose a net~\cite{Mus22} so that, for a suitable constant $c'$, any half-ellipsoid contained within $M$ of volume at least $c' \vol(M)$ contains at least one point of the net. It follows from standard results~\cite{BEHM89,Har11a} that half-ellipsoids have constant VC-dimension, and so the size of the resulting net is $O(1)$. Polytope $P$ is defined to be the convex hull of the points of the nets associated with all the Macbeath regions of $\MM^{1/c}(X)$. 

\begin{figure}[htbp]
    \centering\includegraphics[scale=0.8]{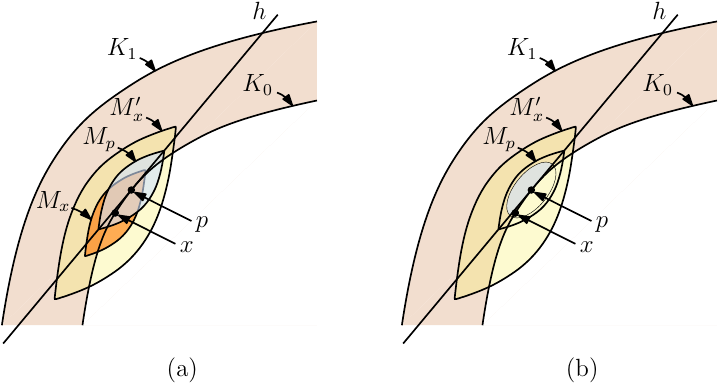}
    \caption{Proof of Lemma~\ref{lem:MNet-approx}.} \label{f:mnet-approx}
\end{figure}

We claim that $K_0 \subset P \subset K_1$. The second containment follows from the fact that the Macbeath regions of $\MM^{1/c}(X)$ are contained within $K_1$. To prove that $K_0 \subset P$, we will show that our construction chooses a point of every cap of $K_1$ induced by a supporting hyperplane of $K_0$. To this end, let $h$ be a supporting hyperplane at some point $p \in \bd K_0$ (see Figure~\ref{f:mnet-approx}(a)), let $H$ be the halfspace bounded by $h$ and not containing $K_0$, and let $C$ be the cap $K_1 \cap H$. Let $M_p = M^{1/4c}(p)$. By the maximality of MNets, there is a point $x \in X$ such that $M_p \cap M_x \neq \emptyset$, where $M_x = M^{1/4c}(x)$. Letting $M'_x = M^{1/c}(x)$ and $M'_p = M^{1/c}(p)$ and applying Lemma~\ref{lem:mac-mac}, we have $M_p \subseteq M'_x$ and $M_x \subseteq M'_p$. Thus, $\vol(M_p) = \Omega(\vol(M'_p)) = \Omega(\vol(M_x))$. By John's Theorem~\cite{Joh48}, $M_p$ contains an ellipsoid $E$ centered at $p$, such that $\vol(E) = \Omega(\vol(M_p))$ (see Figure~\ref{f:mnet-approx}(b)). Combining these observations, it follows that the half-ellipsoid $E' = E \cap H$ has volume $\Omega(\vol(M_x))$. Since $\vol(M'_x) = O(\vol(M_x))$, it follows that a point of the net constructed for $M'_x$ is contained in $E'$ (for a sufficiently small constant $c'$). Noting that $E' \subseteq C$ completes the proof.
\end{proof}

\subsection{Concepts from Projective Geometry} \label{s:projective}

In this section, we present some relevant standard concepts from projective geometry. For further details, see any standard reference (e.g., \cite{Ric11}). Given reals $a,b,c,d \in \RE$, their \emph{cross ratio} $(a,b;c,d)$ is defined to be $(a-c)(b-d) / (a-d)(b-c)$. Given four points on any line $\ell$ in $\RE^d$, we define their cross ratio by identifying $\ell$ with the real line. We follow the convention of using symbols $a,b,c,d,\ldots$ for points, and the distinction from other uses (such as $d$ for the dimension) should be clear from the context.

It is well known that cross ratios are preserved under projective transformations. If the cross ratio $(a,b;c,d)$ is $-1$, we say that this quadruple of points forms a \emph{harmonic bundle} (see Figure~\ref{f:basics}(a)). This is an important special case which occurs frequently in constructions. In this case, the points appear on the line in the order $\ang{a, d, b, c}$ and the ratio in which $a$ divides $c$ and $d$ externally (that is, $(a-c) / (a-d)$) is the same as the ratio in which $b$ divides $c$ and $d$ internally (that is, $(c-b) / (b-d)$). If $a$ is at infinity, it follows easily that $b$ is midway between $c$ and $d$.

\begin{figure}[htbp]
    \centering\includegraphics[scale=0.8]{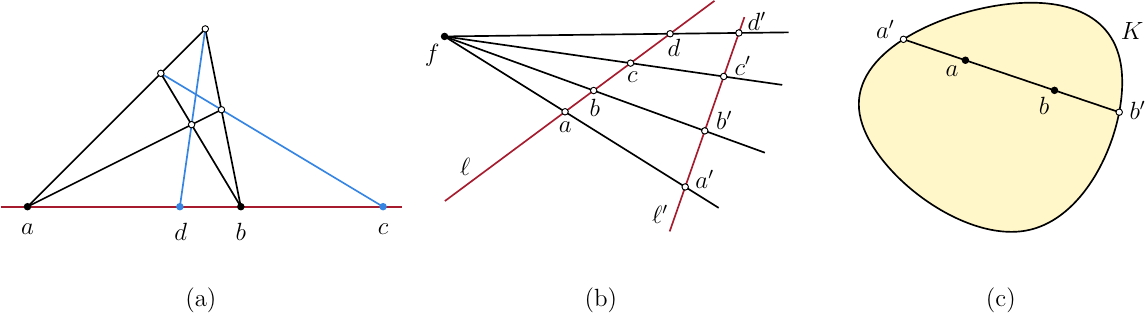}
    \caption{(a) Harmonic bundle, (b) the perspectivity through $f$, and (c) the Hilbert distance.} \label{f:basics}
\end{figure}

Given two lines $\ell$ and $\ell'$ and a point $f$ that lies on neither line, the set of lines through $f$ defines a natural bijection between $\ell$ and $\ell'$, called the \emph{perspectivity through $f$}. Given four points $\{a, b, c, d\}$ on $\ell$ and their respective images $\{a', b', c', d'\}$ on $\ell'$ under this perspectivity, it is well known that the cross ratios $(a, b; c, d)$ and $(a',b';c',d')$ are equal (see Figure~\ref{f:basics}(b)). We express this notationally as
\[
    (a, b; c, d) ~ =_{[f]} ~ (a', b'; c', d').
\]

Every convex body $K$ induces a metric on $\interior(K)$, called the \emph{Hilbert distance}~\cite{Hil95}. Given two points $a, b \in \interior(K)$, let $a'$ and $b'$ denote the points on $\bd K$ intersected by the line through $a$ and $b$, such that they appear in the order $\ang{a', a, b, b'}$ (see Figure~\ref{f:basics}(c)). Then the Hilbert distance between $a$ and $b$, induced by $K$, is defined as $d_K(a,b) = \inv{2} \ln (a,b;b',a')$.

Given $x \in \interior(K)$ and $r \geq 0$, let $B_K(x,r) = \{ y \in K \ST d_K(x,y) \leq r\}$ denote a Hilbert ball of radius $r$ centered at $x$. An important property of Macbeath regions is that they can serve as proxies for Hilbert balls, as shown in the following lemma. (Proofs can be found in~\cite{VeW16} and~\cite{AbM18}.)

\begin{lemma} \label{lem:nesting}
Let $K$ be a convex body and let $x$ be a point in $\interior(K)$. For any $0 < \lambda < 1$,
\[
    B_K\left(x, \frac{1}{2} \ln(1+\lambda)\right) 
        ~ \subseteq ~ M_K^{\lambda}(x) 
        ~ \subseteq ~ B_K\left(x, \frac{1}{2} \ln\left(1 + \frac{2\lambda}{1-\lambda}\right)\right).
\]
\end{lemma}

\subsection{Intermediate Bodies} \label{s:am-hm}

In this section, we explore the concept of relative fatness, which was introduced in Section~\ref{s:techniques}. Given two convex bodies $K_0$ and $K_1$ such that $K_0 \subset K_1$ and $0 < \gamma < 1$, we say that $K_0$ is \emph{relatively $\gamma$-fat} with respect to $K_1$ if, for any point $p \in \bd K_0$, and any scaling factor $0 < \lambda \leq 1$, at least a fraction $\gamma$ of the volume of the Macbeath region $M = M_{K_1}^{\lambda}(p)$ lies within $K_0$, that is, $\vol(M \cap K_0)/\vol(M) \geq \gamma$. 
We say that $K_0$ is \emph{relatively fat} with respect to $K_1$ if it is relatively $\gamma$-fat for some constant $\gamma$. Relative fatness will play an important role in our analyses. Since an arbitrary nested pair $K_0 \subset K_1$ may not necessarily satisfy this property, it will be useful to define an intermediate body sandwiched between $K_0$ and $K_1$ that does.

There are a few natural ways to define such an intermediate body. Given two convex bodies $K_0$ and $K_1$, where $K_0 \subseteq K_1$, the \emph{arithmetic-mean body}, $K_A(K_0, K_1)$, is defined to be the scaled Minkowski sum $\frac{1}{2}(K_0 \oplus K_1)$. Equivalently, for any unit vector $u$ consider the two supporting halfspaces of $K_0$ and $K_1$ orthogonal to $u$, and take the halfspace that is midway between the two (see Figure~\ref{f:arihar}(a)). The arithmetic-mean body is obtained by intersecting such halfspaces for all unit vectors $u$. Clearly, $K_A(K_0, K_1)$ is convex and $K_0 \subseteq K_A(K_0, K_1) \subseteq K_1$.

\begin{figure}[htbp]
    \centering\includegraphics[scale=0.8]{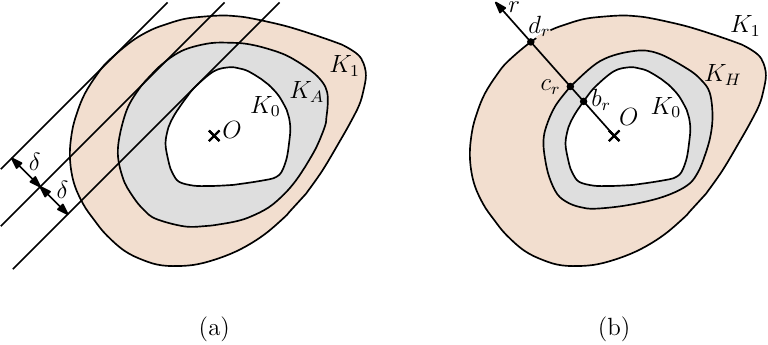}
    \caption{(a) The arithmetic-mean body and (b) the harmonic-mean body.} \label{f:arihar} 
\end{figure}

Another natural choice arises from a polar viewpoint. Assume that $K_0 \subseteq K_1$ and the origin $O \in \interior(K_0)$. The \emph{harmonic-mean body}, $K_H(K_0, K_1)$, was introduced by Firey~\cite{Fir61} and is defined as follows. For any ray $r$ from the origin $O$, let $b_r$ and $d_r$ denote the points of intersection of $r$ with $\bd K_0$ and $\bd K_1$, respectively (see Figure~\ref{f:arihar}(b)). Let $c_r$ be the point on the ray such that 
\[
    \frac{1}{\|O c_r\|} 
        ~ = ~ \frac{1}{2} \left( \frac{1}{\|O b_r\|} + \frac{1}{\|O d_r\|} \right). 
\]
Equivalently, the cross ratio $(O, c_r; d_r, b_r)$ equals $-1$, that is, this quadruple forms a harmonic bundle. Clearly, $c_r$ lies between $b_r$ and $d_r$, and hence the union of these points over all rays $r$ defines the boundary of a body sandwiched between $K_0$ and $K_1$. This body is the harmonic-mean body. By considering the supporting hyperplanes orthogonal to the ray $r$, it is easy to see that the arithmetic-mean body of $K_0$ and $K_1$ is mapped to the harmonic-mean body of $\stdpolar{K}_0$ and $\stdpolar{K}_1$ under polarity, that is, $\stdpolar{(K_A(K_0, K_1))} = K_H(\stdpolar{K}_1, \stdpolar{K}_0)$. It follows that $K_H(K_0, K_1)$ is convex. When $K_0$ and $K_1$ are clear from context, we will just write $K_A$ and $K_H$, omitting references to their arguments. 

\begin{figure}[htbp]
    \centering\includegraphics[scale=0.8,page=2]{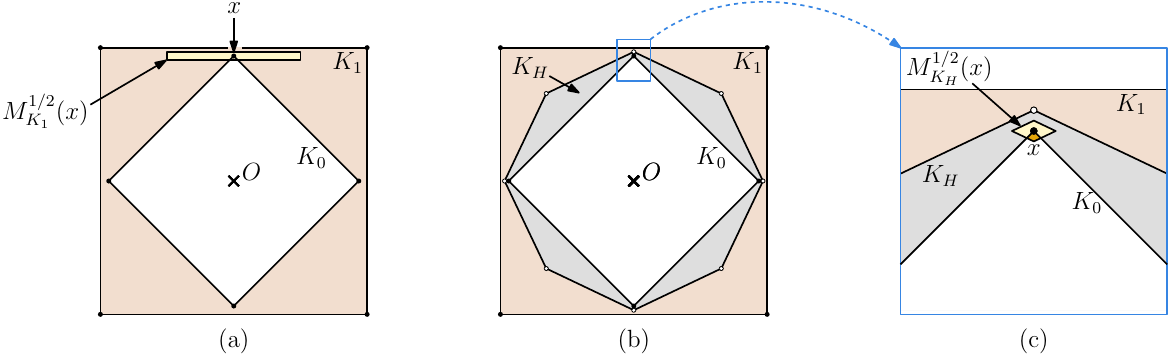}
    \caption{Relative fatness of $K_H$.} \label{f:rel-fat-har} 
\end{figure}

To understand why these intermediate bodies are useful to us, recall the diamond and square bodies $K_0$ and $K_1$ from Figure~\ref{f:rel-fat} (see Figure~\ref{f:rel-fat-har}(a)). Recall the issue that a large fraction of the volume of the Macbeath region $M^{1/2}_{K_1}(x)$ lies outside of $K_0$. If we replace $K_1$ with $K_H = K_H(K_0, K_1)$ and compute the Macbeath region with respect to $K_H$ instead (see Figure~\ref{f:rel-fat-har}(b) and (c)), we see that a constant fraction of the volume of the Macbeath region lies within $K_0$, and so relative fatness is satisfied.

In Section~\ref{s:hm-fat}, we will present an important result by showing that the inner body $K_0$ is relatively fat with respect to the harmonic-mean body $K_H(K_0, K_1)$. The proof makes heavy use of concepts from projective geometry, such as the harmonic bundle. This fact will be critical to establishing the volume-sensitive bounds given in this paper.

\section{Bounding MNet Sizes} \label{s:mnet-sizes}

In this section, we bound the sizes of MNets in important special cases involving points at roughly the same ray distance. These bounds will be useful for obtaining our volume-sensitive bounds. We begin by recalling some definitions and technical tools from~\cite{AFM24}. We say that two caps $C_1$ and $C_2$ are \emph{$\lambda$-similar} for $\lambda \geq 1$, if $C_1 \subseteq C_2^{\lambda}$ and $C_2 \subseteq C_1^{\lambda}$. If two caps are $\lambda$-similar for constant $\lambda$, we say that the caps are \emph{similar}. Note that this is an affine-invariant notion of closeness between caps. 

Arya {\etal}~\cite[Section~2.6 and Section~3.2]{AFM24} showed certain important relationships between caps in $K$ and the associated Macbeath regions in $\stdpolar{K}$. In order to state their result, consider the following mapping. Consider a point $z \in \stdpolar{K}$. Let $\hat{z} \not\in \stdpolar{K}$ be the point on the ray $Oz$ such that $\ray(\hat{z}) = \eps$. The dual hyperplane $\stdpolar{\hat{z}}$ intersects $K$, and so induces a cap, which we call $z$'s \emph{$\eps$-representative cap}. They showed that points that lie within the same shrunken Macbeath regions have similar representative caps, which implies Lemma~\ref{lem:mahler-mac}(i). Furthermore, by extending and generalizing the results in \cite{AAFM22,AFM12b,NNR20}, they established a Mahler-type reciprocal relationship between the volume of caps in $K$ and the associated Macbeath regions in $\stdpolar{K}$. This is stated in Lemma~\ref{lem:mahler-mac}(ii).

\begin{lemma} \label{lem:mahler-mac}
Let $0 < \eps \leq \frac{1}{16}$ and let $K \subseteq \RE^d$ be a well-centered convex body. Let $C$ be a cap of $K$ such that $\eps/2 \leq \width(C) \leq 2\eps$. Suppose that the ray shot from the origin orthogonal to the base of $C$ intersects a Macbeath region $M = M^{1/5}(y)$ of $\stdpolar{K}$, where $\ray(y) = \eps$ (see Figure~\ref{f:sandwich}). Then:
\begin{enumerate}\setlength{\itemsep}{-0.5ex}\setlength{\parsep}{0pt}
\item[$(i)$] The cap $C$ and the $\eps$-representative cap of any point $z \in M$ are 16-similar.
\item[$(ii)$] $\vol_K(C) \cdot \vol_{\stdpolar{K}}(M) = \Omega(\eps^{d+1})$.
\end{enumerate}
\end{lemma}

\begin{figure}[htbp]
    \centering\includegraphics[scale=0.8,page=1]{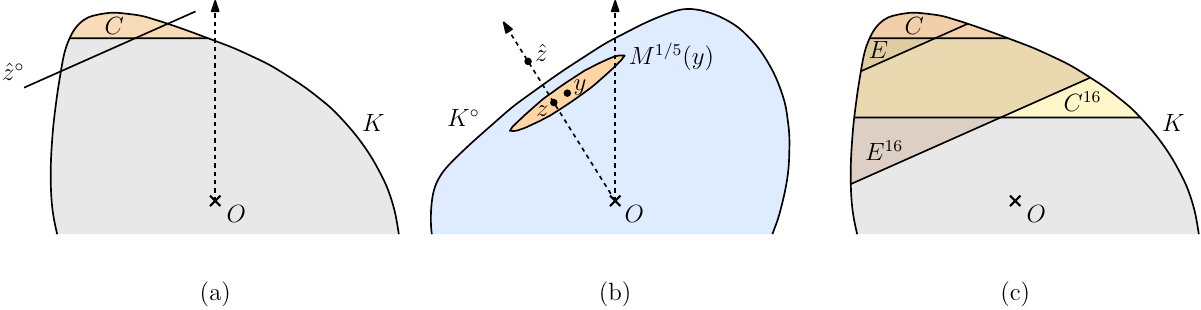}
    \caption{Statement of Lemma~\ref{lem:mahler-mac}(i). Cap $E$ is the $\eps$-representative cap of $z$.} \label{f:sandwich} 
\end{figure}

Next we present a general tool which will be useful in bounding the sizes of the MNets of interest to us. Let $K \subseteq \RE^d$ be a well-centered convex body. For any shallow cap $C$ of $K$, define a point $\psi(C)$ in $\stdpolar{K}$ as follows. In polar space, consider the ray shot from $O$ orthogonal to the base of $C$. We let $\psi(C) \in \stdpolar{K}$ be the point on this ray that has ray distance $\width(C)$.

Let $\CC$ be a set of shallow caps of $K$, let $\Lambda \subseteq K$ denote the set of centroids of the bases of the caps of $\CC$, and let $\Lambda' = \{\psi(C) : C \in \CC\}$. Let $X$ be a $(K,\Lambda)$-MNet, and let $Y$ be a $(\stdpolar{K}, \Lambda')$-MNet. For each $x \in \Lambda$, let $C_x$ denote a cap of $\CC$ such that $x$ is the centroid of its base. (Clearly, such a cap exists. If there is more than one, then we choose one arbitrarily.) Also, for each $x \in X$, define $M_x = M_K^{1/4c}(x)$, where $c$ is the expansion factor of the MNets. Similarly, for $y \in Y$, define $M_y = M_{\stdpolar{K}}^{1/4c}(y)$. The following lemma shows that it is possible to construct a bipartite graph $(X,Y)$ with certain properties. 

\begin{lemma} \label{lem:bipartite} 
Given a well-centered convex body $K \subseteq \RE^d$ and entities $\CC, \Lambda, \Lambda', X, Y$ as defined above, there is a bipartite graph $(X,Y)$ such that there is exactly one edge incident to each vertex of $X$ and the degree of each vertex of $Y$ is $O(1)$. Furthermore, for any $x \in X$ and $y \in Y$, if there is an edge $(x,y)$, then $\vol_K(M_x) \cdot \vol_{\stdpolar{K}}(M_y) = \Omega(\delta^{d+1})$, where $\delta = \width(C_x)$. 
\end{lemma}

\begin{proof}
First, we show how to construct the bipartite graph $(X,Y)$. Let $x$ be any point of $X$ and let $y' = \psi(C_x)$. By the covering property of MNets, there exists $y \in Y$ such that $M^{1/c}(y)$ contains $y'$. We select one such $y$ and add an edge in the bipartite graph between $x$ and $y$. From our construction it follows that there is exactly one edge incident to each vertex of $X$.

\begin{figure}[htbp]
    \centering\includegraphics[scale=0.8,page=2]{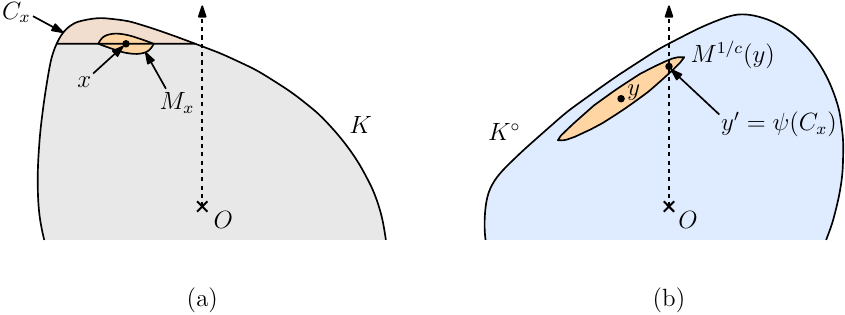}
    \caption{Proof of Lemma~\ref{lem:bipartite}.} \label{f:bipartite} 
\end{figure}

Next, we show that if there is an edge $(x,y)$, then $\vol_K(M_x) \cdot \vol_{\stdpolar{K}}(M_y) = \Omega(\delta^{d+1})$. By definition, $\ray(y') = \width(C_x) = \delta$. Letting $\eps = \ray(y)$ and applying Lemma~\ref{lem:core-ray}, we have $\eps/2 \leq \ray(y') \leq 2 \eps$. Thus, $\eps/2 \leq \width(C_x) \leq 2 \eps$. Observe that the cap $C_x$ and the Macbeath region $M^{1/5}(y)$ satisfy the conditions of Lemma~\ref{lem:mahler-mac}. Recalling that $c$ is a constant, by part (ii) of this lemma, we have $\vol_K(C_x) \cdot \vol_{\stdpolar{K}}(M_y) = \Omega(\eps^{d+1})$. Also, by Lemma~\ref{lem:min-vol-cap2}, $\vol(M_x) = \Omega(\vol(C_x))$. Thus, $\vol_K(M_x) \cdot \vol_{\stdpolar{K}}(M_y) = \Omega(\eps^{d+1})$.

It remains to prove that the degree of each vertex of $Y$ is $O(1)$. Let $y$ be any vertex of $Y$ and let $\eps = \ray(y)$. For any edge $(x,y)$, we showed above that the cap $C_x$ and the Macbeath region $M^{1/5}(y)$ satisfy the conditions of Lemma~\ref{lem:mahler-mac}. By part (i) of this lemma, it follows that the cap $C_x$ and the $\eps$-representative cap of $y$ are $16$-similar. 

Letting $C_y$ denote the $\eps$-representative cap of $y$, we have $C_x \subseteq C_y^{16}$ and $C_y \subseteq C_x^{16}$. Applying Lemma~\ref{lem:cap-exp}, we have $\vol(C_x) = \Omega(\vol(C_x^{16})) = \Omega(\vol(C_y))$, and by Lemma~\ref{lem:min-vol-cap2}, we have $\vol(M_x) = \Omega(\vol(C_x))$. Thus, $\vol(M_x) = \Omega(\vol(C_y))$. Recall that half of the Macbeath region $M_x$ lies within $C_x$, and therefore it lies within $C_y^{16}$. By Lemma~\ref{lem:cap-exp}, $\vol(C_y^{16}) = O(\vol(C_y))$. Since the Macbeath regions of $\MM^{1/4c}(X)$ are disjoint, a straightforward packing argument implies that $y$ has $O(1)$ neighbors.
\end{proof}

Expressing the total number of edges in the graph as the sum of the degrees of the vertices of $Y$, we see that this quantity is $O(|Y|)$. The following corollary is immediate.

\begin{corollary} \label{cor:bipartite} 
Given a convex body $K \subseteq \RE^d$, and the entities $\CC, \Lambda, \Lambda', X, Y$ as defined above, then $|X| = O(|Y|)$. 
\end{corollary}

We are now ready to bound the sizes of MNets in the important special case involving caps of roughly the same width, which map in the polar to points at roughly the same ray distance. Lemmas~\ref{lem:fixed-width} and~\ref{lem:fixed-ray} bound the sizes of the MNets in these cases. We also bound the cardinality of important subsets that arise in our applications. 

\begin{lemma} \label{lem:fixed-width}
Let $0 < \eps \leq \eps_0$, where $\eps_0$ is a sufficiently small constant, and let $K \subseteq \RE^d$ be a well-centered convex body. Let $\CC$ be a set of
caps of $K$ of width between $\eps$ and $2\eps$, let $\Lambda \subseteq K$ denote the set of centroids of the bases of the caps of $\CC$, and let $X$ be a $(K,\Lambda)$-MNet. Then:
\begin{enumerate}
\item[$(i)$] $|X| = O(1/\eps^{(d-1)/2})$.
\item[$(ii)$] For any positive real $f \leq 1$, let $X_f \subseteq X$ be such that the total relative volume of the Macbeath regions of $\MM^{1/4c}(X_f)$ is $O(f \eps)$. Then $|X_f|$ is $O\big( \sqrt{f} / \eps^{(d-1)/2} \big)$.
\end{enumerate}
\end{lemma} 

Note that if $f$ is $o(\eps^{d-1})$, then $\sqrt{f}/\eps^{(d-1)/2}$ is $o(1)$ and so $X_f = \emptyset$.

\begin{proof}
For each point $x \in X$, associate a cap $C_x \in \CC$ such that $x$ is the centroid of its base. Let $M_x = M^{1/4c}(x)$. Let $\Lambda' = \{\psi(C) : C \in \CC\}$, where $\psi$ is as defined above, and let $Y$ be a $(\stdpolar{K},\Lambda')$-MNet. Note that the entities $\CC, \Lambda, \Lambda', X, Y$ satisfy the preconditions of Lemma~\ref{lem:bipartite}.

Arguing as in Lemma~\ref{lem:vol-mac-bounds}, we can show that all the Macbeath regions of $\MM^{1/4c}(X)$ lie in the shell $S_K = K \setminus (1-4\eps)K$, all the Macbeath regions of $\MM^{1/4c}(Y)$ lie in the shell $S_{\stdpolar{K}} = \stdpolar{K} \setminus (1-4\eps)\stdpolar{K}$, $\vol_K(S_K) = O(\eps)$ and $\vol_{\stdpolar{K}}(S_{\stdpolar{K}}) = O(\eps)$.

Define the \emph{fractional volume} of a Macbeath region $M \in \MM^{1/4c}(X)$, denoted $\vol_f(M)$, to be $\vol(M) / \vol(S_K)$. Similarly, for $M \in \MM^{1/4c}(Y)$, define $\vol_f(M) = \vol(M) / \vol(S_{\stdpolar{K}})$. Consider the bipartite graph with vertex sets $X$ and $Y$ described in Lemma~\ref{lem:bipartite}. Recall that there is exactly one edge incident to each vertex of $X$ and the degree of each vertex of $Y$ is $O(1)$.  Further, if there is an edge $(x,y)$, then $\vol_K(M_x) \cdot \vol_{\stdpolar{K}}(M_y) = \Omega(\eps^{d+1})$. Thus
\[
    \vol_f(M_x) \cdot \vol_f(M_y) 
        ~ = ~ \Omega\left(\frac{\vol_K(M_x)}{\vol_K(S_K)} \cdot \frac{\vol_{\stdpolar{K}}(M_y)}{\vol_{\stdpolar{K}}(S_{\stdpolar{K}})} \right) 
        ~ = ~ \Omega\left(\frac{\eps^{d+1}}{\eps \cdot \eps} \right) 
        ~ = ~ \Omega\left(\eps^{d-1}\right).
\]
It follows that the quantity $\vol_f(M_x) + \vol_f(M_y)$ is $\Omega(\eps^{(d-1)/2})$ for any edge $(x,y)$. Summing this quantity over all the edges in the graph, we obtain a lower bound of $\Omega(|X| \, \eps^{(d-1)/2})$. To upper bound this quantity, note that by disjointness, $\sum_{x \in X} \vol_f(M_x) = O(1)$, $\sum_{y \in Y} \vol_f(M_y) = O(1)$, and the degree of each vertex is $O(1)$. Thus, the sum of this quantity over all the edges is $O(1)$. The lower and upper bounds together imply that $|X| = O(1/\eps^{(d-1)/2})$. 

The proof of (ii) is similar to (i). (In fact, (i) is a special case of (ii) for $f=1$.) By Lemma~\ref{lem:vol-mac-bounds}(i), the relative volume of any Macbeath region of $\MM^{1/4c}(X)$ is $\Omega(\eps^d)$. It follows that if $f = o(\eps^{d-1})$ then $X_f = \emptyset$ and so (ii) holds. We may therefore assume that $f = \Omega(\eps^{d-1})$. Letting $S'_K \subseteq S_K$ denote the union of the Macbeath regions of $\MM^{1/4c}(X_f)$, we are given that $\vol_K(S'_K) = O(f \eps)$. We modify the definition of the fractional volume of a Macbeath region $M \in \MM^{1/4c}(X_f)$, denoted $\vol_f(M)$, to be $\vol(M) / \vol(S'_K)$. Note that we keep the same definition of fractional volume for the Macbeath regions of $\MM^{1/4c}(Y)$, that is, for $M \in \MM^{1/4c}(Y)$, $\vol_f(M) = \vol(M) / \vol(S_{\stdpolar{K}})$. Arguing as in (i), but using the bound $\vol_K(S'_K) = O(f \eps)$ in place of  $\vol_K(S_K) = O(\eps)$, it follows that for any edge $(x,y)$ such that $x \in X_f$ and $y \in Y$, we have
\[
\vol_f(M_x) \cdot \vol_f(M_y) 
    ~ = ~ \Omega\left(\frac{\eps^{d-1}}{f} \right).
\]
Thus $\vol_f(M_x) + \vol_f(M_y) = \Omega(\eps^{(d-1)/2}/\sqrt{f})$ for any such edge $(x,y)$. As in (i), summing this quantity over all the edges incident to the vertices of $X_f$, we obtain a lower bound of $\Omega(|X_f| \, \eps^{(d-1)/2}/\sqrt{f})$, and an upper bound of $O(1)$. Together, these bounds imply that $|X_f| = O(\sqrt{f}/\eps^{(d-1)/2})$, as desired.
\end{proof}

The following lemma is analogous to Lemma~\ref{lem:fixed-width}, but for points at similar ray distances. We will use this lemma in Section~\ref{s:hausdorff} together with the relative fatness properties of the harmonic-mean body to establish our volume-sensitive bound. 

\begin{lemma} \label{lem:fixed-ray}
Let $\eps > 0$ and let $K \subseteq \RE^d$ be a well-centered convex body. Let $\Lambda$ be any set of points of $K$ at ray distances between $\eps$ and $2\eps$, and let $X$ be a $(K,\Lambda)$-MNet.
Then:
\begin{enumerate}
    \item[$(i)$] $|X| = O(1/\eps^{(d-1)/2})$.
    \item[$(ii)$] For any positive real $f \leq 1$, let $X_f \subseteq X$ be such that the total relative volume of the Macbeath regions of $\MM^{1/4c}(X_f)$ is $O(f\eps)$. Then $|X_f| = O(\sqrt{f}/\eps^{(d-1)/2})$.
\end{enumerate}
\end{lemma} 

Note that if $f$ is $o(\eps^{d-1})$, then $\sqrt{f}/\eps^{(d-1)/2}$ is $o(1)$ and so $X_f = \emptyset$.

\begin{proof}
Since ray distances are bounded by $1$, we may assume without loss of generality that $\eps \leq 1$, because otherwise $\Lambda$ is empty and the lemma holds vacuously. We associate a minimum volume cap $C_x$ with each point $x \in X$. Recall that $x$ is the centroid of the base of $C_x$. Let $M_x = M^{1/4c}(x)$, and let $\eps_0$ be a sufficiently small constant. By Lemma~\ref{lem:wide-cap}, if the width of $C_x$ exceeds any fixed constant, then $\vol_K(M_x) = \Omega(1)$. Thus, the number of points $x \in X$ such that $\width(C_x) > \eps_0$ is at most $O(1)$. Next, we bound the remaining points of $X$.

Since $\eps \leq \ray(x) \leq 2\eps$, it follows from Lemmas~\ref{lem:raydist-width} and \ref{lem:min-vol-cap3} that $\eps \leq \width(C_x) \leq 2(2d+1)\eps$. Let $\eps_i = 2^i \eps$. We partition the remaining points of $X$ into $O(\log d)$ groups, where group $i$ is denoted $X_i$, such that the widths of the minimum volume caps associated with the points of group $i$ lie between $\eps_i$ and $2 \eps_i$. By Lemma~\ref{lem:fixed-width}(i) (where the caps of the lemma are those associated with the points of $X_i$ and $\Lambda = X_i$), the number of points in group $i$ is $O(1/\eps_i^{(d-1)/2})$. Taking the sum of all the groups $i$, it follows that $|X| = O(1/\eps^{(d-1)/2})$, which proves (i).

The proof of (ii) is similar. By Lemma~\ref{lem:vol-mac-bounds}(ii), the relative volume of any Macbeath region of $M^{1/4c}(X)$ is $\Omega(\eps^d)$. It follows that if $f = o(\eps^{d-1})$ then $X_f = \emptyset$ and so (ii) holds. We may therefore assume that $f = \Omega(\eps^{d-1})$. Arguing as in (i), we can show that the number of points $x \in X_f$ such that $\width(C_x) > \eps_0$ is $O(1)$. We partition the remaining points into $O(\log d)$ groups as before. Applying Lemma~\ref{lem:fixed-width}(ii) to each group and summing the result proves (ii).
\end{proof}

\section{Relative Fatness and the Harmonic-Mean Body} \label{s:hm-fat}

In this section, we establish properties of the harmonic-mean body that are critical to the main results of this paper. In particular, given two bodies $K_0 \subset K_1$, we show that $K_0$ is relatively fat with respect to $K_H$. In fact, we present a stronger result in Lemma~\ref{lem:HM-fat-main}, which implies relative fatness as an immediate consequence. We will employ this stronger result in Section~\ref{s:hausdorff} to obtain our volume-sensitive bounds for polytope approximation.

The proof of Lemma~\ref{lem:HM-fat-main} is based on the following technical lemma. For constant $\lambda$, it implies that for any point $b \in K_0$ that is not too close to the boundary of $K_0$, the Macbeath regions centered at $b$ with respect to $K_0$ and $K_H$, respectively, are roughly similar up to a constant scaling factor. This is established in the following two lemmas. 

\begin{lemma} \label{lem:HM-fat-aux1}
Let $0 < \lambda < 1$ be a parameter. Let $K_0 \subset K_1$ be two convex bodies, where the origin $O$ lies in the interior of $K_0$, and let $K_H$ denote the harmonic-mean body of $K_0$ and $K_1$. Consider any ray emanating from the origin $O$, and let $c$ and $d$ denote the points of intersection of this ray with $\bd K_0$ and $\bd K_1$, respectively (see figure). Let $b \in K_0$ be any point on this ray such that the cross ratio $(O,c; d,b) \leq -\lambda$. Then in the Hilbert metric induced by $K_H$, there is a ball of radius at least $\lambda/6$ centered at $b$ that lies entirely within $K_0$, that is, $B_{K_H}(b,\lambda/6) \subseteq K_0$.
\end{lemma}

\begin{figure}[htbp]
    \centering\includegraphics[scale=0.8]{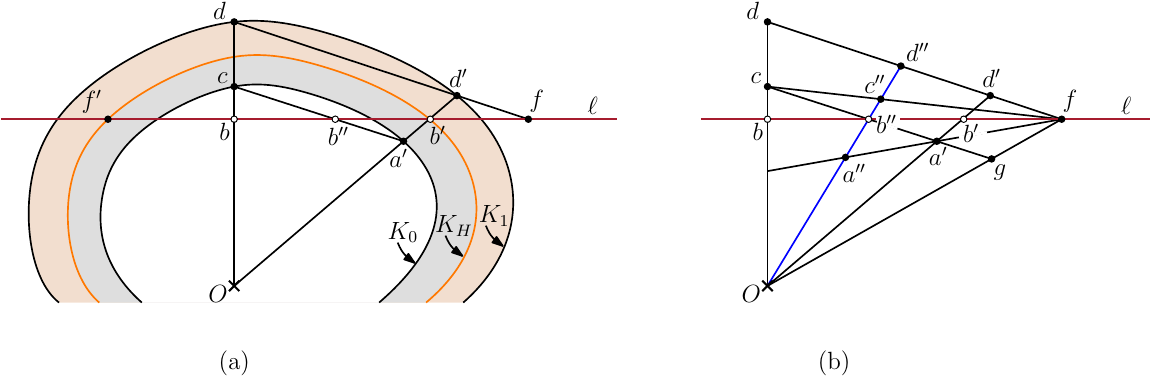}
    \caption{Lemma~\ref{lem:HM-fat-aux1} and its proof.} \label{f:harmonic-mean}
\end{figure}

\begin{proof}
Throughout the proof, we assume that distances are measured in the Hilbert metric induced by $K_H$. To prove the lemma, it suffices to show that for any line $\ell$ passing through $b$, the intersection points of this line with $\bd K_0$ are at distance at least $\lambda/6$ from $b$. We will show this on just one side of $b$, and the other side will follow by symmetry. Consider a ray shot from $b$ along $\ell$, and let $b'$ denote its intersection point with $\bd K_H$ (see Figure~\ref{f:harmonic-mean}(a)). Shoot a ray from $O$ through $b'$, and let $a'$ and $d'$ denote its respective intersections with $\bd K_0$ and $\bd K_1$. Clearly, $c$ and $a'$ lie on opposite sides of $\ell$, and therefore the line segment $c a'$ intersects $\ell$ at a point $b''$, which lies within $K_0$. It suffices to show that $b''$ is at distance at least $\lambda/6$ from $b$.

By the hypotheses of the lemma, we have $(O,c; d,b) \leq -\lambda$. It will simplify the proof to assume that it is in fact equal to $-\lambda$. Since $b'$ lies on $\bd K_H$, we have $(O,b'; d',a') = -1$. Letting $f'$ denote the intersection of $\ell$ and $\bd K_H$ on the other side of $b$, the Hilbert distance between $b$ and $b''$ is
\[
    d_{K_H}(b, b'')
        ~ = ~ \frac{1}{2} \ln (b, b''; b', f').
\]
Let $f$ denote the intersection of the line $d d'$ with $\ell$. (Note that it may lie on either side of $b$.) Clearly, $f$ is exterior to $K_1$, and therefore it does not lie within the segment $f' b'$. It follows that the cross ratio can only become smaller by using $f$ instead of $f'$, which implies that
\[
    d_{K_H}(b, b'')
        ~ \geq ~ \frac{1}{2} \ln (b, b''; b', f).
\]

The remainder of the proof involves relating these various cross ratios. First, let us bring them all onto a common line. Consider the ray from $O$ through $b''$, and let $a''$, $c''$, and $d''$ denote its respective intersections with the lines $f a'$, $f c$, and $f d$, and let $g$ denote the intersection of line $c a'$ with $O f$ (see Figure~\ref{f:harmonic-mean}(b)). By applying perspectivities through $f$ (recall Section~\ref{s:projective}), we have          
\begin{equation}
    -1 ~ = ~ (O,b'; d',a') ~ =_{[f]} ~ (O,b''; d'',a'')
    \qquad\text{and}\qquad
    -\lambda ~ = ~ (O,c; d,b) ~ =_{[f]} ~ (O,c''; d'',b''). \label{eq:HM-fat-aux1-1}
\end{equation}
Also, by applying first a perspectivity through $O$ and then another through $f$, we have
\[
    (b, b''; b', f)
        ~ =_{[O]} ~ (c,b''; a',g)
        ~ =_{[f]} ~ (c'',b''; a'',O).
\]

Taking the ratios of the items in Eq.~\eqref{eq:HM-fat-aux1-1}, we have
\[
    \frac{1}{\lambda}
        ~ = ~ \frac{(O,b''; d'',a'')}{(O,c''; d'',b'')}.
\]
By expanding the definition of cross ratio and simple algebra, it is easily verified that
\[
    \frac{1}{\lambda}
        ~ = ~ \frac{(O,b''; d'',a'')}{(O,c''; d'',b'')} 
        ~ = ~ -\frac{(O,b''; c'',a'')}{(O,d''; c'',b'')}.
\]
By standard identities on cross ratios \cite[Theorem 4.2]{Ric11}, $(O,d''; c'',b'') = 1 - (O,c''; d'',b'')$ and $(O,b'';c'',a'') = 1/(1 - (c'',b''; a'',O))$. Therefore
\[
    \frac{1}{\lambda}
        ~ = ~ -\frac{(O,b''; c'',a'')}{1 - (O,c''; d'',b'')}
        ~ = ~ -\frac{1/(1 - (c'',b''; a'',O))}{1 + \lambda}.
\]
Solving for $(c'',b''; a'',O)$, yields
\[
    (c'',b''; a'',O)
        ~ = ~ 1 + \frac{\lambda}{1+\lambda}.
\]
Using the fact that $0 < \lambda < 1$ and $\ln (1+x) \geq x/(x+1)$, we conclude that
\[
    d_{K_H}(b, b'')
        ~ \geq ~ \frac{1}{2} \ln (b, b''; b', f)
        ~ =    ~ \frac{1}{2} \ln (c'',b''; a'',O)
        ~ =    ~ \frac{1}{2} \ln \left( 1 + \frac{\lambda}{1+\lambda} \right)
        ~ \geq ~ \frac{\lambda}{2(1+2\lambda)}
        ~ \geq ~ \frac{\lambda}{6},
\]
as desired.
\end{proof}

\begin{lemma} \label{lem:HM-fat-aux2}
Given the entities defined in the statement of Lemma~\ref{lem:HM-fat-aux1}, $M_{K_H}^{s(\lambda)}(b) \subseteq M_{K_0}(b)$, where $s(\lambda) = \lambda / 8$.
\end{lemma}

\begin{proof}
By Lemma~\ref{lem:nesting}, for any convex body $K$, $x \in K$, and $0 < \alpha < 1$, 
\[
    M_K^{\alpha}(x) ~ \subseteq ~ B_{K}\left(x, \frac{1}{2} \ln \left(1 + \frac{2 \alpha}{1 - \alpha}\right)\right).
\]
From Lemma~\ref{lem:HM-fat-aux1}, $K_0$ contains the Hilbert ball (with respect to $K_H$) of radius $\lambda/6$ centered at $b$. Setting $\alpha = \lambda/8$, we have $\alpha \leq 1/8$ and therefore 
\[
    1 + \frac{2 \alpha}{1 - \alpha}
        ~ \leq ~ 1 + \frac{2 \lambda/8}{7/8}
        ~ =    ~ 1 + \frac{2\lambda}{7}.
\]
Thus, by setting $s(\lambda) = \alpha = \lambda/8$ and applying the inequality $\ln (1+x) \leq x$ for $x > 0$, it follows that 
\[
    M_{K_H}^{s(\lambda)}(b) 
        ~ \subseteq ~ B_{K_H}\left(b, \frac{1}{2} \ln \left(1 + \frac{2 \lambda}{7}\right)\right) 
        ~ \subseteq ~ B_{K_H}\left(b, \frac{1}{2} \cdot \frac{2\lambda}{7}\right) 
        ~ \subseteq ~ B_{K_H}\left(b, \frac{\lambda}{6}\right)    ~ \subseteq ~ K_0.
\]
Since $M_{K_H}^{s(\lambda)}(b)$ is centrally symmetric, it also lies within $M_{K_0}(b)$.
\end{proof}

We have the following lemma which, in conjunction with Lemma~\ref{lem:HM-fat-aux2}, will be useful in proving Lemma~\ref{lem:HM-fat-main}.

\begin{lemma} \label{lem:cr-lb}
Let $\lambda, K_0, K_1, K_H$, the origin $O$, and points $c$ and $d$ be as in Lemma~\ref{lem:HM-fat-aux1}. Let $h$ denote the point of intersection of ray $Oc$ with the boundary of $K_H$. Then:
\begin{enumerate}
\item[$(i)$] $\|Oc\| \geq \|hc\|$. 
\item[$(ii)$] Let $b$ be a point on the segment $Oc$ that is not contained in the interior of $M_{K_H}^{\lambda}(c)$. Then $(O,c;d,b) \leq -\lambda / 2$.
\end{enumerate}
\end{lemma}

\begin{proof}
Let $h$ and $h'$ denote the intersection points of the line $Oc$ with the boundary of $K_H$ such that $h', b, c,$ and $h$ appear in this order on the line. Since $(O,h;d,c)$ forms a harmonic bundle, we have $\|Od\|/\|Oc\| =\|hd\|/|hc\|$. Since $\|Od\| \geq \|hd\|$, it follows that $\|Oc\| \geq \|hc\|$, which proves (i).

To prove (ii), observe that since $b$ is not contained in the interior of $M_{K_H}^{\lambda}(c)$, Lemma~\ref{lem:nesting} implies that it is not contained in $B_{K_H}\big( c, \frac{1}{2} \ln (1+\lambda) \big)$. By the definition of the Hilbert distance, we have $(b,c;h,h') \geq 1+\lambda$. By standard cross-ratio identities \cite[Theorem 4.2]{Ric11} $(b,c;h,h') = 1 - (h',c;h,b)$, and so $(h',c;h,b) \leq -\lambda$. Replacing $h'$ by $O$ can only decrease the cross ratio and hence $(O,c;h,b) \leq -\lambda$. 

\begin{figure}[htbp]
    \centering\includegraphics[scale=0.8]{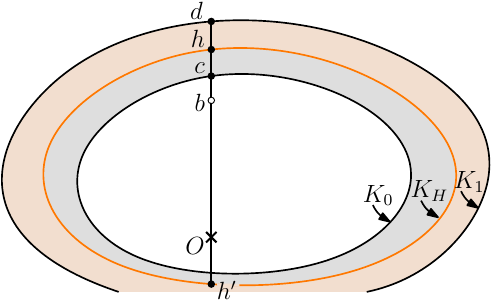}
    \caption{Proof of Lemma~\ref{lem:cr-lb}.} \label{f:cross-ratio-bound}
\end{figure}

To complete the proof, it suffices to show that $(O,c;d,b) = (O,c;h,b) / 2$. We apply a projective transformation that maps $O$ to a point at infinity. (Recall that collinearities and cross ratios are preserved under projective transformations.) Since $O$ is at infinity, its contributions to the cross ratio cancel out, and we have 
\[
    (O,c;d,b) 
        ~ = ~ -\frac{\|cb\|}{\|cd\|}, \qquad
    (O,c;h,b) 
        ~ = ~ -\frac{\|cb\|}{\|ch\|}, \text{~~~and~~~} 
    (O,h;d,c) 
        ~ = ~ -\frac{\|ch\|}{\|hd\|}.
\]
Also, according to the definition of the harmonic-mean body, $(O,h;d,c)$ is a harmonic bundle, and so $\|ch\| = \|hd\|$. Thus,
\[
    (O,c;d,b) 
        ~ = ~ -\frac{\|cb\|}{\|cd\|} 
        ~ = ~ -\frac{\|cb\|}{\|ch\| + \|hd\|} 
        ~ = ~ -\frac{\|cb\|}{2\|ch\|} 
        ~ = ~ \frac{1}{2} \SP (O,c;h,b),
\]
as desired.
\end{proof}

We now have all the key ingredients to present the main result of this section. The relative fatness of $K_0$ with respect to $K_H$ is an immediate consequence of parts (i) and (ii) of this lemma. In order to state part (iii), we need a definition. Given a convex body $K$ with the origin $O$ in its interior and a region $R \subseteq K$, define the \emph{shadow} of $R$ with respect to $K$, denoted \emph{$\shadow_K(R)$}, to be the set of points $x \in K$ such that the segment $Ox$ intersects $R$. 

\begin{lemma} \label{lem:HM-fat-main}
Let $0 < \beta \leq 1$ be a real parameter. Let $K_0 \subset K_1$ be two convex bodies, let the origin $O$ lie in the interior of $K_0$, and let $K_H$ denote the harmonic-mean body of $K_0$ and $K_1$. Let $c$ be any point on the boundary of $K_0$, and let $M = M_{K_H}^{\beta}(c)$. Then there exists a convex body $M'$ such that 
\begin{enumerate}\setlength{\itemsep}{-0.5ex}\setlength{\parsep}{0pt}
\item[$(i)$] $\vol(M') = \Omega(\vol(M))$,
\item[$(ii)$] $M' \subseteq M \cap K_0$, and
\item[$(iii)$] $\shadow_{K_0}(M') \subseteq M$.
\end{enumerate}
\end{lemma}

\begin{proof} 
For the sake of convenience, assume that the ray $Oc$ is directed vertically upwards. Let $h$ be the point of intersection of the ray $Oc$ with $\bd K_H$. Let $R = M_{K_H}(c)-c$ be the recentering of $M_{K_H}(c)$ about the origin. By definition, $M = M_{K_H}^{\beta}(c) = c + \beta R$. Let $b$ be the point of intersection of the segment $Oc$ with the boundary of $M_{K_H}^{\lambda}(c) = c + \lambda R$, where $\lambda = \beta/\kappa$ for a suitable large constant $\kappa \geq 2$ (independent of dimension). Recalling from Lemma~\ref{lem:cr-lb}(i) that $\|ch\| \leq \|Oc\|$, it follows that $b$ is vertically below $c$ at a distance of $\lambda \|ch\|$. Recalling $s(\lambda)$ from Lemma~\ref{lem:HM-fat-aux2}, let $M' = b + \gamma R$ for
\[
    \gamma 
        ~ = ~ \frac{s(\lambda/2)}{10} 
        ~ = ~ \frac{s(\beta/2\kappa)}{10}  
        ~ = ~ \frac{\beta}{160\kappa}
\]
(see Figure~\ref{f:hm-fat-main}(a)). Since $M'$ and $M$ are translated copies of $R$ scaled by a factor of $\gamma$ and $\beta$, respectively, we have $\vol(M') = (\gamma/\beta)^d \vol(M) = (1/160\kappa)^d \vol(M)$. This proves (i).

\begin{figure}[htbp]
    \centering\includegraphics[scale=0.8]{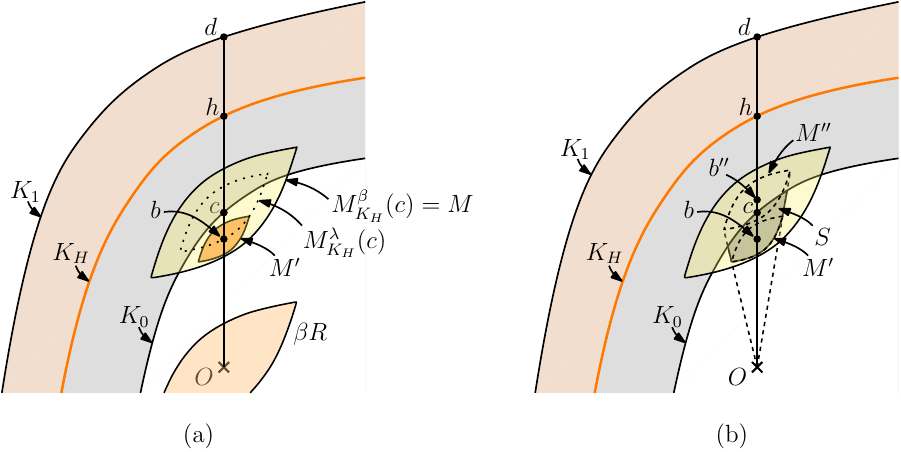}
    \caption{Proof of Lemma~\ref{lem:HM-fat-main}. (Objects are not drawn to scale.)} \label{f:hm-fat-main} 
\end{figure}

To prove (ii), we will show that $M' \subseteq M$ and $M' \subseteq K_0$. Since $b \in c + \lambda R$ and $M' = b + \gamma R$, it follows that $M' \subseteq c + (\lambda + \gamma) R$. For large $\kappa$, we have $\lambda + \gamma \leq \beta$, and thus $M' \subseteq c + \beta R = M$. 

Next, we show that $M' \subseteq K_0$. Let $d$ denote the point of intersection of the ray $Oc$ with $\bd K_1$. Applying Lemma~\ref{lem:cr-lb}(ii), it follows that the cross ratio $(O,c;d,b) \leq -\lambda / 2$. Applying Lemma~\ref{lem:HM-fat-aux2} with $\lambda/2$ in place of $\lambda$ and recalling that $s(\lambda/2) = 10\gamma$, we have $M_{K_H}^{10 \gamma}(b) \subseteq M_{K_0}(b)$. Also, by Lemma~\ref{lem:mac-trans}, we have $M' = b + \gamma R \subseteq M_{K_H}^{2\gamma}(b)$. Thus $M' \subseteq M_{K_0}^{1/5}(b)$. By definition of Macbeath regions, $M_{K_0}(b) \subseteq K_0$, and so $M' \subseteq K_0$, as desired. 

To prove (iii), let $S = \shadow_{K_0}(M')$, and let $M''$ be the convex body obtained by scaling $M'$ by the factor 
\[
    f 
        ~ = ~ 1 + 4 \lambda \frac{\|ch\|}{\|Oc\|}
\]
about $O$ (see Figure~\ref{f:hm-fat-main}(b)). Letting $b''$ denote the center of $M''$, we have $M'' = b'' + f \gamma R$. We claim that
\begin{enumerate}
\item[(a)] $S$ is contained in the convex hull of $M' \cup M''$, and

\item[(b)] the convex hull of $M' \cup M''$ is contained in $M$.
\end{enumerate} 
Together, this would imply that $S$ is contained in $M$, and complete the proof.

To prove (a), let $c'$ be any point in $S \cap \bd K_0$ and let $b'$ be any point in the intersection of segment $Oc'$ with $M'$. Since $b' \in M'$ and $M' \subseteq M_{K_0}^{1/5}(b)$, we have $b' \in M_{K_0}^{1/5}(b)$. By Lemma~\ref{lem:core-ray}, we have $\ray_{K_0}(b') \leq 2 \cdot \ray_{K_0}(b)$, that is, $\|b'c'\| / \|Oc'\| \leq 2 \|bc\| / \|Oc\|$. Recalling that $\|bc\| = \lambda \|ch\|$, it follows that
\[
    \frac{\|Oc'\|}{\|Ob'\|} 
        ~ =    ~ \frac{\|Oc'\|}{\|Oc'\| - \|b'c'\|} 
        ~ =    ~ \frac{1}{1 - \frac{\|b'c'\|}{\|Oc'\|}} 
        ~ \leq ~ \frac{1}{1 - 2 \lambda \frac{\|ch\|}{\|Oc\|}} 
        ~ \leq ~ 1 + 4 \lambda \frac{\|ch\|}{\|Oc\|}.
\]
Recall that we defined the quantity on the right hand side to be the scaling factor $f$ and $M''$ to be the $f$-factor expansion of $M'$ about $O$. Since $\|Oc'\| \leq f \|Ob'\|$, it follows that for any ray passing through $M'$, the points of $S$ on this ray lie between the lowest point on the ray contained in $M'$ and the highest point on the ray contained in $M''$. It follows that $S$ is contained in the convex hull of $M' \cup M''$.

It remains to prove (b).  By convexity of $M$, it suffices to show that both $M'$ and $M''$ are contained in $M$. We have already shown in part (ii) that $M' \subseteq M$. To complete the proof, it suffices to show that $M'' \subseteq M$. Note that the point corresponding to $c$, obtained by scaling by a factor of $f$ about the origin, is at distance $4 \lambda \|ch\|$ vertically above $c$. Clearly $b''$ lies on the segment $bc''$, and so $b'' \in c + 4 \lambda R$. Recalling that $M''= b'' + f \gamma R$, it follows that $M'' \subseteq c + (4 \lambda + f \gamma) R$. Since $\lambda = \beta/\kappa$, and $\gamma = \beta/160\kappa$, and $f = 1 + 4\lambda \|ch\|/\|Oc\| \leq 1 + 4 \lambda$, we have $4\lambda + f\gamma \leq \beta$, for large $\kappa$. Thus $M'' \subseteq c + \beta R = M$, which completes the proof.
\end{proof}

The following corollary follows immediately from parts (i) and (ii) of the above lemma.

\begin{corollary}
Let $K_0 \subset K_1$ be two convex bodies, let the origin $O$ lie in the interior of $K_0$, and let $K_H$ denote the harmonic-mean body of $K_0$ and $K_1$. Then $K_0$ is relatively fat with respect to $K_H$.    
\end{corollary}

\section{Uniform Volume-Sensitive Bounds} \label{s:hausdorff}

In this section, we prove the paper's main result, Theorem~\ref{thm:main}. Let $\eps > 0$ and let $K_0$ denote the convex body $K$ described in this theorem. Let $K_1 = K_0 \oplus \eps B^d_2$ denote the Minkowski sum of $K_0$ with a Euclidean ball of radius $\eps$ (see Figure~\ref{f:hausdorff}(a)). Also, recall that $\Delta_d(K_0)$ denotes the \emph{volume diameter} of $K_0$. Let $m(K_0,\eps)$ be a shorthand for $(\Delta_d(K_0)/\eps)^{(d-1)/2}$, the desired number of facets.

We will show that there exists a polytope with $O(m(K_0,\eps))$ facets sandwiched between $K_0$ and $K_1$. As mentioned above, we will transform the problem by mapping it to the polar. Through an appropriate translation, we may assume that the centroid of $K_0$ coincides with the origin $O$. Note that the arithmetic-mean body $K_A$ of $K_0$ and $K_1$ is given by $K_0 \oplus \frac{\eps}{2} B^d_2$, and recall from Section~\ref{s:am-hm} that $K_H = \stdpolar{K}_A$ is the harmonic-mean body of 
$\stdpolar{K}_1$ and $\stdpolar{K}_0$ (see Figure~\ref{f:hausdorff}(b)). 

\begin{figure}[htbp]
    \centering\includegraphics[scale=0.8]{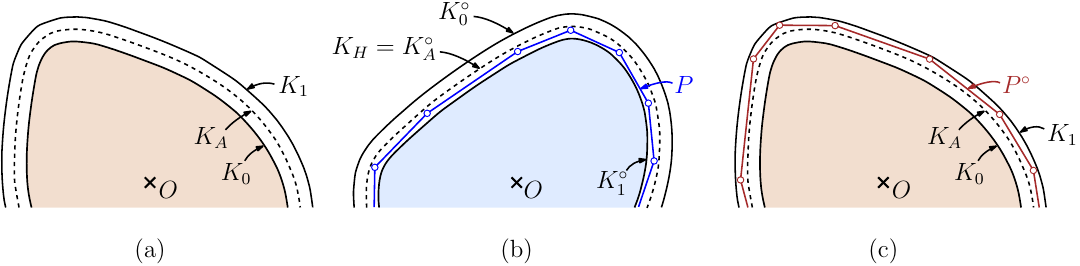}
    \caption{Volume-sensitive Hausdorff approximation.} \label{f:hausdorff}
\end{figure}

Our construction is based on Lemma~\ref{lem:MNet-size-hausdorff} below, which shows that there is a $(K_H,\bd (\stdpolar{K}_1))$-MNet $X$ of size $O(m(K_0,\eps))$. By applying Lemma~\ref{lem:MNet-approx}, it follows that there exists a polytope $P$ sandwiched between $\stdpolar{K}_1$ and $K_H$ with $O(|X|)$ vertices (see Figure~\ref{f:hausdorff}(b)). By polarity, this implies that $\stdpolar{P}$ is a polytope sandwiched between $K_A$ and $K_1$ having $O(|X|)$ facets (see Figure~\ref{f:hausdorff}(c)). Since $K_0 \subseteq K_A$, this polytope is also sandwiched between $K_0$ and $K_1$, which proves Theorem~\ref{thm:main}.

All that remains is to show that $|X| = O(m(K_0,\eps))$. For this purpose, we will utilize the tools for bounding the sizes of MNets in conjunction with the relative fatness of the harmonic-mean body (established in Section~\ref{s:hm-fat}).

\begin{lemma} \label{lem:MNet-size-hausdorff}
Let $\eps > 0$ and let $K_0, K_1, K_A, K_H$ be convex bodies as defined above. Let $X$ be a $(K_H,\bd (\stdpolar{K}_1))$-MNet. Then $|X| = O(m(K_0,\eps))$.
\end{lemma}

\begin{proof}
We begin by showing that $\vol(K_H) = \Omega(1/\vol(K_0))$, and its Mahler volume is at most $O(1)$ (implying that $K_H$ is well-centered). To see this, recall that the width of $K_0$ in any direction is at least $\eps$ and $K_A = K_0 \oplus \frac{\eps}{2} B^d_2$. It is well known that for any convex body, the ratio of the distances of the body's centroid from any pair of supporting hyperplanes is at most $d$ (see, e.g., \cite{Gru63}). It follows that a ball of radius $\eps/(d+1)$ centered at the origin (the centroid of $K_0$) lies within $K_0$. Thus, a constant-factor expansion of $K_0$ contains $K_A$, which implies that $\vol(K_A) = O(\vol(K_0))$. Also, because $K_H = \stdpolar{K}_A$, by Lemma~\ref{lem:mahler-bounds}(i), $\vol(K_A) \cdot \vol(K_H) = \Omega(1)$. Thus, $\vol(K_H) = \Omega(1/\vol(K_0))$. To upper bound the Mahler volume of $K_H$, note that by polarity, $K_H \subseteq \stdpolar{K}_0$, and thus
\[
    \vol(K_H) \cdot \vol(\stdpolar{K}_H) 
        ~ = ~ \vol(K_A) \cdot \vol(K_H) 
        ~ = ~ O(\vol(K_0) \cdot \vol(\stdpolar{K}_0))
        ~ = ~ O(1),
\] 
where in the last step, we have used the fact that the centroid of $K_0$ coincides with the origin and Lemma~\ref{lem:mahler-bounds}(ii).

To simplify notation, for the remainder of the proof we assume that ray distances, Macbeath regions, and volumes are defined relative to $K_H$, that is, $\ray \equiv \ray_{K_H}$, $M \equiv M_{K_H}$, and $\vol \equiv \vol_{K_H}$.

For any point $p \in \bd (\stdpolar{K}_1)$, let $p'$ denote the point of intersection of the ray $O p$ with $\bd K_H$. We first establish a bound on the relative ray distance $\ray(p)$. Observe that since $p$ and $p'$ lie on $\bd (\stdpolar{K}_1)$ and $\bd K_H$, respectively, their polar hyperplanes, $\stdpolar{p}$ and $\stdpolar{p'}$, are supporting hyperplanes for $K_1$ and $\stdpolar{K}_H = K_A$, respectively. Letting $r$ denote the distance between $\stdpolar{p'}$ and the origin, it follows from the definition of $K_A$ that the distance between $\stdpolar{p}$ and the origin is $r + \frac{\eps}{2}$. The distance of $p'$ and $p$ from the origin are the reciprocals of these. Therefore, we have
\[
    \ray(p)
        ~ = ~ \frac{\|p p'\|}{\|O p'\|}
        ~ = ~ \frac{\|O p'\| - \|O p\|}{\|O p'\|}
        ~ = ~ \frac{\frac{1}{r} - \frac{1}{r + (\eps/2)}}{\frac{1}{r}}
        ~ = ~ 1 - \frac{r}{r + (\eps/2)}
        ~ = ~ \frac{\eps/2}{r + (\eps/2)}.
\]
Since $1/\|O p'\| = r = \Omega(\eps)$, we have $\ray(p) = \Theta(\eps/r) = \Theta(\eps \|O p'\|)$. (It is noteworthy and somewhat surprising that this relative ray distance is not a dimensionless quantity, since it depends linearly on $\|O p'\|$.)

To analyze $|X|$, we partition it into groups based on $\|O x'\|$ for each $x \in X$, where $x' $ denotes the point of intersection of the ray $O x$ with $\bd K_H$. Define $R_0 = (\vol(K_H))^{1/d}$. By our earlier remarks, $\vol(K_H) = \Omega(1/\vol(K_0))$, and so $R_0 = \Omega(1/\Delta_d(K_0))$. For any integer $i$ (possibly negative), define $R_i = 2^i R_0$ and $\eps_i = \eps R_i$. We can express $X$ as the disjoint union of sets $X_i$, where $X_i$ consists of points $x$ such that $R_i \leq \|Ox'\| < 2 R_i$. Recall that for any $x \in X_i$, we have $\ray(x) = \Theta(\eps \|O x'\|) = \Theta(\eps R_i) = \Theta(\eps_i)$. 

We will bound the contributions of the sizes of each $X_i$ to the size of $X$ based on the sign of $i$. Let us first consider the nonnegative values of $i$. Recalling that $K_H$ is well-centered and applying Lemma~\ref{lem:fixed-ray}(i) (where $X_i$ takes the role of $\Lambda$ in the lemma), we have, up to constant factors
\begin{align*}
    \sum_{i \geq 0} |X_i|
        & ~ \leq ~ \sum_{i \geq 0} \left( \frac{1}{\eps_i}\right)^{\kern-2pt\frac{d-1}{2}}
          ~ =    ~ \sum_{i \geq 0} \left( \frac{1}{\eps 2^i R_0}\right)^{\kern-2pt\frac{d-1}{2}}
          ~ \leq ~ \sum_{i \geq 0} \left( \frac{\Delta_d(K_0)}{\eps 2^i}\right)^{\kern-2pt\frac{d-1}{2}} \\
        & ~ =    ~ \left( \frac{\Delta_d(K_0)}{\eps} \right)^{\kern-2pt\frac{d-1}{2}}  \sum_{i \geq 0} \left(\frac{1}{2}\right)^{\kern-2pt\frac{i(d-1)}{2}}        
          ~ \leq ~ \left( \frac{\Delta_d(K_0)}{\eps} \right)^{\kern-2pt\frac{d-1}{2}} 
          ~ =    ~ O(m(K_0,\eps)),
\end{align*}
as desired.

In order to bound the contributions to $|X|$ for negative values of $i$, we need a more sophisticated strategy. Our approach is to first bound the total relative volume of the Macbeath regions of $\MM^{1/4c}(X_i)$, which we assert to be $O(\eps_i 2^{id})$. Assuming this assertion for now, we complete the proof as follows. By applying Lemma~\ref{lem:fixed-ray}(ii)  with $f_i = O(2^{id})$ and recalling that $\eps_i = \eps R_i = 2^i \eps R_0$, we have (up to constant factors)
\begin{align*}
    \sum_{i < 0} |X_i| 
        & ~ \leq ~ \sum_{i < 0} \frac{\sqrt{f_i}}{\eps_i^{(d-1)/2}} 
          ~ =    ~ \sum_{i < 0} \frac{2^{i d/2}}{(2^i \eps R_0)^{(d-1)/2}}
          ~ =    ~ \sum_{i < 0} \frac{2^{i(d-(d-1))/2}}{(\eps R_0)^{(d-1)/2}}
          ~ =    ~ \sum_{i < 0} \frac{2^{i/2}}{(\eps R_0)^{(d-1)/2}} \\
        & ~ \leq ~ \sum_{i < 0} 2^{i/2} m(K_0,\eps)
          ~ =    ~ m(K_0,\eps) \sum_{i > 0} \left( \frac{1}{2} \right)^{\kern-1pt\frac{i}{2}}
          ~ =    ~ O(m(K_0, \eps)),
\end{align*}
as desired.

It remains only to prove the assertion on the total relative volume of the Macbeath regions of $\MM^{1/4c}(X_i)$. Let $x \in X_i$ and let $M_x = M^{1/4c}(x)$. By Lemma~\ref{lem:HM-fat-main} (with $x$, $\stdpolar{K}_1$, and $K_H$ playing the roles of $c$, $K_0$, and $K_H$, respectively), there is an associated convex body $M'_x$ such that 
\begin{center}
    (i) $\vol(M'_x) = \Omega(\vol(M_x))$, ~~
    (ii) $M'_x \subseteq M_x \cap \stdpolar{K}_1$, ~~and~~ 
    (iii) $\shadow_{\stdpolar{K}_1}(M'_x) \subseteq M_x$.
\end{center}
We will use $S_x$ as a shorthand for $\shadow_{\stdpolar{K}_1}(M'_x)$.  Since $\vol(M_x) = O(\vol(M'_x)) = O(\vol(S_x))$, it suffices to show that the total relative volume of the shadows $\{S_x : x \in X_i\}$ is $O(\eps_i 2^{id})$.

For $x \in X_i$, we define the cone $\Psi_x$ to be the intersection of $K_H$ with the infinite cone consisting of rays emanating from the origin that contain a point of $S_x$ (see Figure~\ref{f:mnet-size-hausdorff}). Since the Macbeath regions of $\MM^{1/4c}(X_i)$ are disjoint, it follows from (iii) that the associated shadows intersect $\bd (\stdpolar{K}_1)$ in patches that are also disjoint. Thus, the cones $\Psi = \{\Psi_x : x \in X_i\}$ are pairwise disjoint.

\begin{figure}[htbp]
    \centering\includegraphics[scale=0.8]{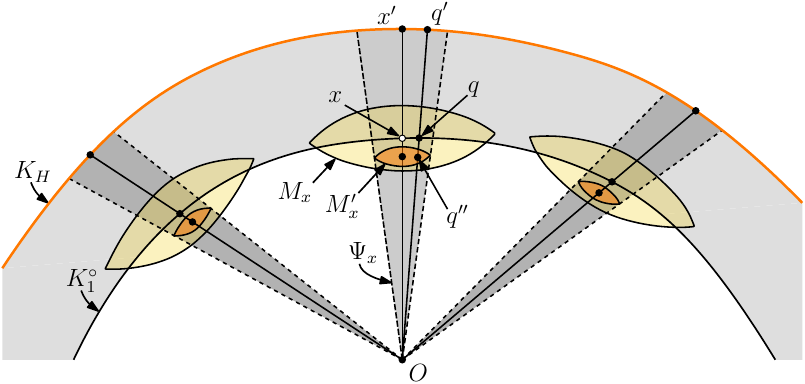}
    \caption{Proof of Lemma~\ref{lem:MNet-size-hausdorff}.} \label{f:mnet-size-hausdorff}
\end{figure}

Consider a ray emanating from the origin that is contained in any cone $\Psi_x$. Let $q$ and $q'$ be the intersection points of this ray with $\bd (\stdpolar{K}_1)$ and $\bd K_H$, respectively. Let $q''$ be any point on this ray that lies within the shadow $S_x$. Since $q'' \in M_x$, by Lemma~\ref{lem:core-ray}, we have $\ray(q'') = \Theta(\ray(x)) = \Theta(\eps_i)$. By the same reasoning, $\ray(q) = \Theta(\eps_i) = \Theta(\eps R_i)$. Also, recalling our earlier bounds on the relative ray distance of points on $\bd (\stdpolar{K}_1)$, we have $\ray(q) = \Theta(\eps \|Oq'\|)$. Equating the two expressions for $\ray(q)$, we obtain $\|Oq'\| = \Theta(R_i)$.

Since the cones of $\Psi$ are disjoint and any ray emanating from the origin and contained in a cone of $\Psi$ has length $\Theta(R_i)$, it follows that the total volume of these cones is $O(R_i^d)$. Furthermore, since only a fraction $O(\eps_i)$ of such a ray is contained in the associated shadow, it follows that the total volume of all the shadows $\{S_x : x \in X_i\}$ is $O(\eps_i R_i^d)$. Recalling that $\vol(K_H) = R_0^d$ and $R_i = 2^i R_0$, it follows that the total relative volume of these shadows is $O(\eps_i R_i^d / R_0^d) = O(\eps_i 2^{id})$. This establishes the assertion on the total relative volume of the Macbeath regions of $\MM^{1/4c}(X_i)$ and completes the proof.
\end{proof}

\section{Nonuniform Volume-Sensitive Bounds} \label{s:nonuniform}

In this section, we observe that a nonuniform bound very similar to ours can be derived from a result due to Gruber \cite{Gru93a}. He proved that if $K$ is a strictly convex body and $\partial K$ is twice differentiable ($C^2$ continuous), then there exists a constant $k_d$ (depending only on the dimension $d$) and $\eps_0$ depending on $K$, such that for any $0 < \eps \leq \eps_0$, the number of bounding halfspaces needed to achieve an $\eps$-approximation of $K$ is at most
\begin{equation} \label{eq:gruber}
    k_d \left( \frac{1}{\eps} \right)^{\kern-2pt \frac{d-1}{2}} \int_{\bd K} \kappa(x)^{\frac{1}{2}} d\sigma(x), 
\end{equation}
where $\kappa$ and $\sigma$ denote the Gaussian curvature of $K$ and ordinary surface area measure, respectively. (B{\" o}r{\" o}czky showed that the requirement that $K$ be ``strictly'' convex can be eliminated \cite{Bor00}.)

\begin{theorem} \label{thm:nonunif-bound}
For any integer $d \geq 2$ and any convex body $K \subseteq \RE^d$ whose boundary is $C^2$ smooth, there exists $\eps_0$ depending on $K$, such that for any $0 < \eps \leq \eps_0$, there exists an $\eps$-approximating polytope $P$ having at most
\[
    c_d \left(\frac{\Delta_d(K)}{\eps}\right)^{\frac{d-1}{2}}
\]
facets, where $c_d$ is a constant (depending on $d$).
\end{theorem}

\begin{proof}
Assume that the centroid of $K$ coincides with the origin. Let $S^{d-1}$ denote the unit Euclidean sphere in $\RE^d$. For $u \in S^{d-1}$, let $h(u) = \max \, \{ \inner{x}{u} : x \in K\}$ denote the support function of $K$ and let $\rho(u) = \max \, \{\lambda > 0: \lambda u \in \stdpolar{K}\}$ denote the radial function of $\stdpolar{K}$. 

Letting $n$ denote the exterior normal unit vector of $K$ and applying the Cauchy-Schwarz inequality, we obtain
\[
    \left(\int_{\bd K} \kappa(x)^{1/2} \, d\sigma(x) \right)^2 
    ~ \leq ~ \left( \int_{\bd K} \frac{\kappa(x)}{h(n(x))} \, d\sigma(x) \right) \cdot
             \left(\int_{\bd K} h(n(x)) \, d\sigma(x)\right).
\]
The second integral on the right-hand side is easily seen to be $d \cdot \vol(K)$. To bound the first integral on the right-hand side, we express it as an integral over the unit sphere $S^{d-1}$.
\[
    \int_{\bd K} \frac{\kappa(x)}{h(n(x))} \, d\sigma(x) 
        ~ = ~ \int_{S^{d-1}} \frac{1}{h(u)} \, d\sigma(u).
\]
Letting $\varsigma_{d-1} = \area(S^{d-1})$ and applying Jensen's inequality, we have
\[
    \frac{1}{\varsigma_{d-1}} \int_{S^{d-1}} \frac{1}{h(u)} \, d\sigma(u)
        ~ \leq ~ \left(\frac{1}{\varsigma_{d-1}} \int_{S^{d-1}} \frac{1}{h(u)^d} \, d\sigma(u)\right)^{\kern-2pt \frac{1}{d}}
        ~ = ~ \left(\frac{1}{\varsigma_{d-1}} \int_{S^{d-1}} \rho(u)^d \, d\sigma(u)\right)^{\kern-2pt \frac{1}{d}},
\]
where we have used the polar relationship $\rho(u) = 1/h(u)$ in the last step. It is easy to see that this integral is $d \cdot \vol(\stdpolar{K})$. Neglecting constant factors depending on $d$, we have thus shown that the first integral on the right-hand side of Eq.~\eqref{eq:gruber} is $O(\vol(\stdpolar{K})^{1/d})$. Thus,
\[
    \left(\int_{\bd K} \kappa(x)^{\frac{1}{2}} \, d\sigma(x) \right)^2 
        ~ = ~ O\left( \vol(\stdpolar{K})^{\frac{1}{d}} \cdot \vol(K) \right)
        ~ = ~ O\left( \vol(K)^{1-\frac{1}{d}} \right),
\]
where we have used Lemma~\ref{lem:mahler-bounds}(ii). Substituting this into Eq.~\eqref{eq:gruber} and recalling that the volume diameter of $K$, $\Delta_d(K) = \Theta((\vol(K))^{1/d})$, implies that the number of bounding halfspaces needed to achieve an $\eps$-approximation of $K$ is at most
\[
	c_d \left(\frac{\Delta_d(K)}{\eps} \right)^{\kern-2pt\frac{d-1}{2}}.
\]
as desired.
\end{proof}

Note that the bound in this theorem matches the uniform bound of Theorem~\ref{thm:main}. However, this approach cannot be used to produce a uniform bound, since Eq.~\eqref{eq:gruber} only holds when $\eps \leq \eps_0$, where $\eps_0$ depends on $K$. (See~\cite{AFM12b} for a counterexample showing that this equation could be violated otherwise.)

\section*{Conclusions}

In this paper, we have shown that it is possible to obtain improved bounds on the complexity of approximating convex bodies with respect to the Hausdorff distance when the skinniness of the body is taken into consideration. Given a convex body $K$ in $\RE^d$, we characterize its skinniness in terms of its volume diameter, $\Delta_d(K)$, defined as the diameter of a Euclidean ball of the same volume as $K$. Our bounds are a substantial improvement over diameter-based bounds by Dudley~\cite{Dud74}, and a marginal improvement over the area-sensitive bounds of Arya, da Fonseca, and Mount~\cite{AFM12b}. We showed that as a function of volume alone, our bound is tight up to constant factors that depend on the dimension.

Although our approach follows earlier work by using covers based on Macbeath regions, we introduced a number of new ideas to deal with skinny bodies. In particular, we introduced the notion of relative fatness for two convex bodies, where one is nested within the other. Here, the fatness of the inner convex body is measured with respect to how tightly it fits within the outer body. We also introduced two intermediate bodies, the arithmetic-mean body and the harmonic-mean body. We showed that the inner body is relatively fat with respect to the harmonic-mean body. We believe that these concepts may be useful in other applications of convex approximation.

Although our bound is asymptotically tight when skinniness is described in terms of intrinsic volumes, it may be far from optimal for any given instance. A major open problem is obtaining an efficient approximation construction that is optimal to within constant factors for any given instance. Another issue is the complexity of the analysis. While the construction given in this paper is relatively simple (involving covering the boundary of a convex body by shrunken Macbeath regions), the analysis is quite technical and involved. An important next step would be to obtain a simpler analysis. There are examples in the literature (see, e.g., \cite{GSW24}) of simpler approximation analyses that are similar in structure to ours.

Our use of covers in both the primal and polar settings raises a number of interesting questions involving the complexity of such covers. Given two nested convex bodies $K_0 \subset K_1$, both of which contain the origin in their interior, we know that $\stdpolar{K}_1 \subset \stdpolar{K}_0$. It is natural to conjecture that the size of a $(K_1, K_0)$-MNet and the size of a $(\stdpolar{K}_0, \stdpolar{K}_1)$-MNet should be related to each other. Faifman showed that the Holmes-Thompson volumes of $K_0$ with respect to $K_1$ and the Holmes-Thompson volume of $\stdpolar{K}_1$ with respect to $\stdpolar{K}_0$ in the Hilbert geometry are equal up to constant factors~\cite{Fai24}. Known similarities between Macbeath regions and Hilbert balls (see, e.g., \cite{VeW16} and~\cite{AbM18}) suggest that a similar relationship may exist for such complementary sets of MNets as well.

\section*{Acknowledgments}

The authors acknowledge the insights and feedback from Rahul Arya and Guilherme da Fonseca. We also thank the anonymous reviewers for their feedback and suggestions.

\pdfbookmark[1]{References}{s:ref}
\bibliographystyle{plainurl}
\bibliography{shortcuts,convex}

\begin{thebibliography}{10}

\bibitem{AbM18}
A.~Abdelkader and D.~M. Mount.
\newblock Economical {Delone} sets for approximating convex bodies.
\newblock In {\em Proc.\ 16th Scand.\ Workshop Algorithm Theory}, pages
  4:1--4:12, 2018.
\newblock \href {https://doi.org/10.4230/LIPIcs.SWAT.2018.4}
  {\path{doi:10.4230/LIPIcs.SWAT.2018.4}}.

\bibitem{AHV05}
P.~K. Agarwal, S.~Har-Peled, and K.~R. Varadarajan.
\newblock Geometric approximation via coresets.
\newblock In J.~E. Goodman, J.~Pach, and E.~Welzl, editors, {\em Combinatorial
  and Computational Geometry}. MSRI Publications, 2005.

\bibitem{AAFM22}
R.~Arya, S.~Arya, G.~D. da~Fonseca, and D.~M. Mount.
\newblock Optimal bound on the combinatorial complexity of approximating
  polytopes.
\newblock {\em ACM Trans.\ Algorithms}, 18:1--29, 2022.
\newblock \href {https://doi.org/10.1145/3559106} {\path{doi:10.1145/3559106}}.

\bibitem{ArC14}
S.~Arya and T.~M. Chan.
\newblock Better $\varepsilon$-dependencies for offline approximate nearest
  neighbor search, {Euclidean} minimum spanning trees, and
  $\varepsilon$-kernels.
\newblock In {\em Proc.\ 30th Annu.\ Sympos.\ Comput.\ Geom.}, pages 416--425,
  2014.
\newblock \href {https://doi.org/10.1145/2582112.2582161}
  {\path{doi:10.1145/2582112.2582161}}.

\bibitem{AFM12b}
S.~Arya, G.~D. da~Fonseca, and D.~M. Mount.
\newblock Optimal area-sensitive bounds for polytope approximation.
\newblock In {\em Proc.\ 28th Annu.\ Sympos.\ Comput.\ Geom.}, pages 363--372,
  2012.
\newblock \href {https://doi.org/10.1145/2261250.2261305}
  {\path{doi:10.1145/2261250.2261305}}.

\bibitem{AFM17b}
S.~Arya, G.~D. da~Fonseca, and D.~M. Mount.
\newblock Near-optimal $\varepsilon$-kernel construction and related problems.
\newblock In {\em Proc.\ 33rd Internat.\ Sympos.\ Comput.\ Geom.}, pages
  10:1--15, 2017.
\newblock URL: \url{https://arxiv.org/abs/1703.10868}, \href
  {https://doi.org/10.4230/LIPIcs.SoCG.2017.10}
  {\path{doi:10.4230/LIPIcs.SoCG.2017.10}}.

\bibitem{AFM17a}
S.~Arya, G.~D. da~Fonseca, and D.~M. Mount.
\newblock Optimal approximate polytope membership.
\newblock In {\em Proc.\ 28th Annu.\ ACM-SIAM Sympos.\ Discrete Algorithms},
  pages 270--288, 2017.
\newblock \href {https://doi.org/10.1137/1.9781611974782.18}
  {\path{doi:10.1137/1.9781611974782.18}}.

\bibitem{AFM24}
S.~Arya, G.~D. da~Fonseca, and D.~M. Mount.
\newblock Economical convex coverings and applications.
\newblock {\em SIAM J.\ Comput.}, 53(4):1002--1038, 2024.
\newblock \href {https://doi.org/10.1137/23M1568351}
  {\path{doi:10.1137/23M1568351}}.

\bibitem{Bor00}
K.~{B{\" o}r{\" o}czky Jr.}
\newblock Approximation of general smooth convex bodies.
\newblock {\em Adv.\ Math.}, 153:325--341, 2000.
\newblock \href {https://doi.org/10.1006/aima.1999.1904}
  {\path{doi:10.1006/aima.1999.1904}}.

\bibitem{Bar00}
I.~B{\'a}r{\'a}ny.
\newblock The technique of {M}-regions and cap-coverings: {A} survey.
\newblock {\em Rend.\ Circ.\ Mat.\ Palermo}, 65:21--38, 2000.
\newblock URL: \url{https://users.renyi.hu/~barany/}.

\bibitem{Bar07}
I.~B{\'a}r{\'a}ny.
\newblock Random polytopes, convex bodies, and approximation.
\newblock In W.~Weil, editor, {\em Stochastic Geometry}, volume 1892 of {\em
  Lecture Notes in Mathematics}, pages 77--118. Springer, 2007.
\newblock \href {https://doi.org/10.1007/978-3-540-38175-4_2}
  {\path{doi:10.1007/978-3-540-38175-4_2}}.

\bibitem{BEHM89}
A.~Blumer, A.~Ehrenfeucht, D.~Haussler, and M.~K. Warmuth.
\newblock Learnability and the {Vapnik-Chervonenkis} dimension.
\newblock {\em J.\ Assoc.\ Comput.\ Mach.}, 36(4):929--965, 1989.
\newblock \href {https://doi.org/10.1145/76359.76371}
  {\path{doi:10.1145/76359.76371}}.

\bibitem{Bon18}
G.~Bonnet.
\newblock Polytopal approximation of elongated convex bodies.
\newblock {\em Advances in Geometry}, 18:105--114, 2018.
\newblock \href {https://doi.org/10.1515/advgeom-2017-0038}
  {\path{doi:10.1515/advgeom-2017-0038}}.

\bibitem{BoM87}
J.~Bourgain and V.~D. Milman.
\newblock New volume ratio properties for convex symmetric bodies.
\newblock {\em Invent.\ Math.}, 88:319--340, 1987.
\newblock \href {https://doi.org/10.1007/BF01388911}
  {\path{doi:10.1007/BF01388911}}.

\bibitem{BCP93}
H.~Br{\"o}nnimann, B.~Chazelle, and J.~Pach.
\newblock How hard is halfspace range searching?
\newblock {\em Discrete Comput.\ Geom.}, 10:143--155, 1993.
\newblock \href {https://doi.org/10.1007/BF02573971}
  {\path{doi:10.1007/BF02573971}}.

\bibitem{BrI76}
E.~M. Bronshteyn and L.~D. Ivanov.
\newblock The approximation of convex sets by polyhedra.
\newblock {\em Siberian Math.\ J.}, 16:852--853, 1976.
\newblock \href {https://doi.org/10.1007/BF00967115}
  {\path{doi:10.1007/BF00967115}}.

\bibitem{Bro08}
E.~M. Bronstein.
\newblock Approximation of convex sets by polytopes.
\newblock {\em J.\ Math.\ Sci.}, 153(6):727--762, 2008.
\newblock \href {https://doi.org/10.1007/s10958-008-9144-x}
  {\path{doi:10.1007/s10958-008-9144-x}}.

\bibitem{Cla06}
K.~L. Clarkson.
\newblock Building triangulations using $\varepsilon$-nets.
\newblock In {\em Proc.\ 38th Annu.\ ACM Sympos.\ Theory Comput.}, pages
  326--335, 2006.
\newblock \href {https://doi.org/10.1145/1132516.1132564}
  {\path{doi:10.1145/1132516.1132564}}.

\bibitem{Dud74}
R.~M. Dudley.
\newblock Metric entropy of some classes of sets with differentiable
  boundaries.
\newblock {\em J.\ Approx.\ Theory}, 10(3):227--236, 1974.
\newblock \href {https://doi.org/10.1016/0021-9045(74)90120-8}
  {\path{doi:10.1016/0021-9045(74)90120-8}}.

\bibitem{Egg58}
H.~G. Eggleston.
\newblock {\em Convexity}.
\newblock Cambridge Univ.\ Press, 1958.
\newblock \href {https://doi.org/10.1017/CBO9780511566172}
  {\path{doi:10.1017/CBO9780511566172}}.

\bibitem{ELR70}
G.~Ewald, D.~G. Larman, and C.~A. Rogers.
\newblock The directions of the line segments and of the $r$-dimensional balls
  on the boundary of a convex body in {Euclidean} space.
\newblock {\em Mathematika}, 17:1--20, 1970.
\newblock \href {https://doi.org/10.1112/S0025579300002655}
  {\path{doi:10.1112/S0025579300002655}}.

\bibitem{Fai24}
D.~Faifman.
\newblock A {Funk} perspective on billiards, projective geometry and {Mahler}
  volume.
\newblock {\em J.\ Differential Geom.}, 127:161--212, 2024.
\newblock \href {https://doi.org/10.4310/jdg/1717356157}
  {\path{doi:10.4310/jdg/1717356157}}.

\bibitem{Fir61}
W.~J. Firey.
\newblock Polar means of convex bodies and a dual to the {Brunn-Minkowski}
  theorem.
\newblock {\em Canad.\ J.\ Math}, 13:444--453, 1961.
\newblock \href {https://doi.org/10.4153/CJM-1961-037-0}
  {\path{doi:10.4153/CJM-1961-037-0}}.

\bibitem{Gru93a}
P.~M. Gruber.
\newblock Aspects of approximation of convex bodies.
\newblock In P.~M. Gruber and J.~M. Wills, editors, {\em Handbook of Convex
  Geometry}, chapter 1.10, pages 319--345. North-Holland, 1993.
\newblock \href {https://doi.org/10.1016/B978-0-444-89596-7.50015-8}
  {\path{doi:10.1016/B978-0-444-89596-7.50015-8}}.

\bibitem{Gru63}
B.~Gr{\"u}nbaum.
\newblock Measures of symmetry for convex sets.
\newblock In {\em Proc.\ Sympos. Pure Math.}, volume VII, pages 233--270, 1963.
\newblock \href {https://doi.org/10.1090/pspum/007/0156259}
  {\path{doi:10.1090/pspum/007/0156259}}.

\bibitem{GSW24}
J.~Gudmundsson, M.~P. Seybold, and S.~Wong.
\newblock Approximating multiplicatively weighted {Voronoi} diagrams: Efficient
  construction with linear size.
\newblock In {\em Proc.\ 40th Internat.\ Sympos.\ Comput.\ Geom.}, pages
  62:1--62:14, 2024.
\newblock \href {https://doi.org/10.4230/LIPIcs.SoCG.2024.62}
  {\path{doi:10.4230/LIPIcs.SoCG.2024.62}}.

\bibitem{Har11a}
S.~Har-Peled.
\newblock {\em Geometric Approximation Algorithms}, volume 173.
\newblock American Mathematical Society, 2011.
\newblock \href {https://doi.org/10.1090/surv/173}
  {\path{doi:10.1090/surv/173}}.

\bibitem{HaJ24}
S.~Har-Peled and M.~Jones.
\newblock Proof of {Dudley's} convex approximation, 2024.
\newblock \href {https://arxiv.org/abs/1912.01977} {\path{arXiv:1912.01977}}.

\bibitem{Hil95}
D.~Hilbert.
\newblock Ueber die gerade {Linie} als k{\" u}rzeste {Verbindung} zweier
  {Punkte}.
\newblock {\em Mathematische Annalen}, 46:91--96, 1895.
\newblock \href {https://doi.org/10.1007/BF02096204}
  {\path{doi:10.1007/BF02096204}}.

\bibitem{Joh48}
F.~John.
\newblock Extremum problems with inequalities as subsidiary conditions.
\newblock In {\em Studies and Essays Presented to R. Courant on his 60th
  Birthday}, pages 187--204. Interscience Publishers, Inc., New York, 1948.

\bibitem{Kup08}
G.~Kuperberg.
\newblock From the {Mahler} conjecture to {Gauss} linking integrals.
\newblock {\em Geom.\ Funct.\ Anal.}, 18:870--892, 2008.
\newblock \href {https://doi.org/10.1007/s00039-008-0669-4}
  {\path{doi:10.1007/s00039-008-0669-4}}.

\bibitem{McV75}
D.~E. McClure and R.~A. Vitalie.
\newblock Polygonal approximation of plane convex bodies.
\newblock {\em J.\ Math.\ Anal.\ Appl.}, 51:326--358, 1975.
\newblock \href {https://doi.org/10.1016/0022-247X(75)90125-0}
  {\path{doi:10.1016/0022-247X(75)90125-0}}.

\bibitem{McM75}
P.~McMullen.
\newblock Non-linear angle-sum relations for polyhedral cones and polytopes.
\newblock {\em Math.\ Proc.\ Cambridge Philos.\ Soc}, 78(2):247--261, 1975.
\newblock \href {https://doi.org/10.1017/S0305004100051665}
  {\path{doi:10.1017/S0305004100051665}}.

\bibitem{McM91}
P.~McMullen.
\newblock Inequalities between intrinsic volumes.
\newblock {\em Monatshefte f{\" u}r Mathematik}, 111:47--53, 1991.
\newblock \href {https://doi.org/10.1007/BF01299276}
  {\path{doi:10.1007/BF01299276}}.

\bibitem{MeP90}
M.~Meyer and A.~Pajor.
\newblock On the {Blaschke}-{Santal{\' o}} inequality.
\newblock {\em Arch.\ Math.}, 55:82--93, 1990.
\newblock \href {https://doi.org/10.1007/BF01199119}
  {\path{doi:10.1007/BF01199119}}.

\bibitem{Mus22}
N.~H. Mustafa.
\newblock {\em Sampling in Combinatorial and Geometric Set Systems}, volume 265
  of {\em Mathematical Surveys and Monographs}.
\newblock AMS, 2022.
\newblock \href {https://doi.org/10.1090/surv/265}
  {\path{doi:10.1090/surv/265}}.

\bibitem{NNR20}
M.~Nasz{\' o}di, F.~Nazarov, and D.~Ryabogin.
\newblock Fine approximation of convex bodies by polytopes.
\newblock {\em Amer.\ J.\ Math}, 142:809--820, 2020.
\newblock \href {https://doi.org/10.1353/ajm.2020.0018}
  {\path{doi:10.1353/ajm.2020.0018}}.

\bibitem{Naz12}
F.~Nazarov.
\newblock The {H\"o}rmander proof of the {Bourgain-Milman} theorem.
\newblock In {\em Geometric Aspects of Functional Analysis}, pages 335--343.
  Springer, 2012.
\newblock \href {https://doi.org/10.1007/978-3-642-29849-3_20}
  {\path{doi:10.1007/978-3-642-29849-3_20}}.

\bibitem{Ric11}
J.~Richter-Gebert.
\newblock {\em Perspectives on Projective Geometry: {A} Guided Tour Through
  Real and Complex Geometry}.
\newblock Springer, 2011.
\newblock \href {https://doi.org/10.1007/978-3-642-17286-1}
  {\path{doi:10.1007/978-3-642-17286-1}}.

\bibitem{San49}
L.~A. Santal{\' o}.
\newblock An affine invariant for convex bodies of $n$-dimensional space.
\newblock {\em Port.\ Math.}, 8:155--161, 1949.
\newblock (In Spanish).

\bibitem{Sch87}
R.~Schneider.
\newblock Polyhedral approximation of smooth convex bodies.
\newblock {\em J.\ Math.\ Anal.\ Appl.}, 128:470--474, 1987.
\newblock \href {https://doi.org/10.1016/0022-247X(87)90197-1}
  {\path{doi:10.1016/0022-247X(87)90197-1}}.

\bibitem{Sch93}
R.~Schneider.
\newblock {\em Convex bodies: {The} {Brunn-Minkowski} theory}.
\newblock Cambridge University Press, 1993.
\newblock \href {https://doi.org/10.1017/CBO9781139003858}
  {\path{doi:10.1017/CBO9781139003858}}.

\bibitem{Tot48}
L.~F. Toth.
\newblock Approximation by polygons and polyhedra.
\newblock {\em Bull.\ Amer.\ Math.\ Soc.}, 54:431--438, 1948.
\newblock \href {https://doi.org/10.1090/S0002-9904-1948-09022-X}
  {\path{doi:10.1090/S0002-9904-1948-09022-X}}.

\bibitem{VeW16}
C.~Vernicos and C.~Walsh.
\newblock Flag-approximability of convex bodies and volume growth of {Hilbert}
  geometries.
\newblock HAL Archive (hal-01423693i), 2016.
\newblock URL: \url{https://hal.archives-ouvertes.fr/hal-01423693}.

\end{thebibliography}

\end{document}